\newif\ifreport\reporttrue
\documentclass[journal,10pt,twocolumns]{IEEEtran}

\usepackage{graphicx}
\usepackage[center]{caption}
\usepackage{epsfig,latexsym,amsmath,amsfonts}
\usepackage{amssymb}
\usepackage{placeins}
\usepackage{subfigure}
\usepackage{verbatim}
\usepackage{cite,url}
\usepackage{epsf,bm}
\usepackage{color}	
\usepackage{epstopdf}
\usepackage{multirow}
\usepackage[onehalfspacing]{setspace}
\DeclareMathOperator*{\argmax}{argmax}
\usepackage{dsfont}
\usepackage[utf8]{inputenc}
\usepackage[english]{babel}
\usepackage{amsthm}
\theoremstyle{definition}
\newtheorem{definition}{Definition}[section]
\newtheorem{remark}{Remark}
\usepackage[lined, boxed, linesnumbered, ruled]{algorithm2e}

\usepackage[utf8]{inputenc}
\usepackage[english]{babel}

\newtheorem{theorem}{Theorem}
\newtheorem{lemma}{Lemma}

\newtheorem{corollary}{Corollary}

\begin{document}
\title{The Age of Information in Multihop Networks}

\author{\large Ahmed M. Bedewy, Yin Sun, \emph{Member, IEEE}, and Ness B. Shroff, \emph{Fellow, IEEE}
\thanks{This paper was presented in part at IEEE ISIT 2017 \cite{Bedewy_Multihop_conf}.}
\thanks {This work has been supported in part by ONR grants N00014-17-1-2417 and N00014-15-1-2166, Army Research Office grants W911NF-14-1-0368 and MURI W911NF-12-1-0385,  National Science Foundation grants CNS-1446582, CNS-1421576, CNS-1518829, and CCF-1813050, and a grant from the Defense Thrust Reduction Agency HDTRA1-14-1-0058.}
\thanks{A. M. Bedewy is with the  Department  of  ECE,  The  Ohio  State  University, Columbus, OH 43210 USA (e-mail:  bedewy.2@osu.edu).}
\thanks{Y.  Sun  is  with  the  Department  of  ECE,  Auburn  University,  Auburn,  AL 36849 USA (e-mail:  yzs0078@auburn.edu).}
\thanks{N. B.  Shroff  is  with  the  Department  of  ECE and  the  Department  of  CSE, The Ohio State University,  Columbus, OH 43210 USA  (e-mail:  shroff.11@osu.edu).}
}
\maketitle

\pagenumbering{arabic}
\begin{abstract}
Information updates in multihop networks such as Internet of Things (IoT) and intelligent transportation systems have received significant recent attention. In this paper, we minimize the  age of a single information flow in interference-free multihop networks. When preemption is allowed and the packet transmission times are exponentially distributed, we prove that a preemptive Last-Generated, First-Served (LGFS) policy results in smaller age processes across all nodes in the network  than any other causal policy (in a stochastic ordering sense). In addition, for the class of New-Better-than-Used (NBU) distributions, we show that the non-preemptive LGFS policy is within a constant age gap from the optimum average age. In contrast, our numerical result shows that the preemptive LGFS policy can be very far from the optimum for some NBU transmission time distributions. Finally,  when preemption is prohibited and the packet transmission times are arbitrarily distributed, the non-preemptive LGFS policy is shown to minimize the age processes across all nodes in the network among all work-conserving policies (again in a stochastic ordering sense). Interestingly, these results hold under quite general conditions, including (i) arbitrary packet generation and arrival times, and (ii) for minimizing both the age processes in stochastic ordering and any non-decreasing functional of the age processes.

\end{abstract}
\section{Introduction}\label{Int}
 \begin{figure*}[!tbp]
 \centering
 \subfigure[Single-gateway model]{
  \includegraphics[scale=0.65]{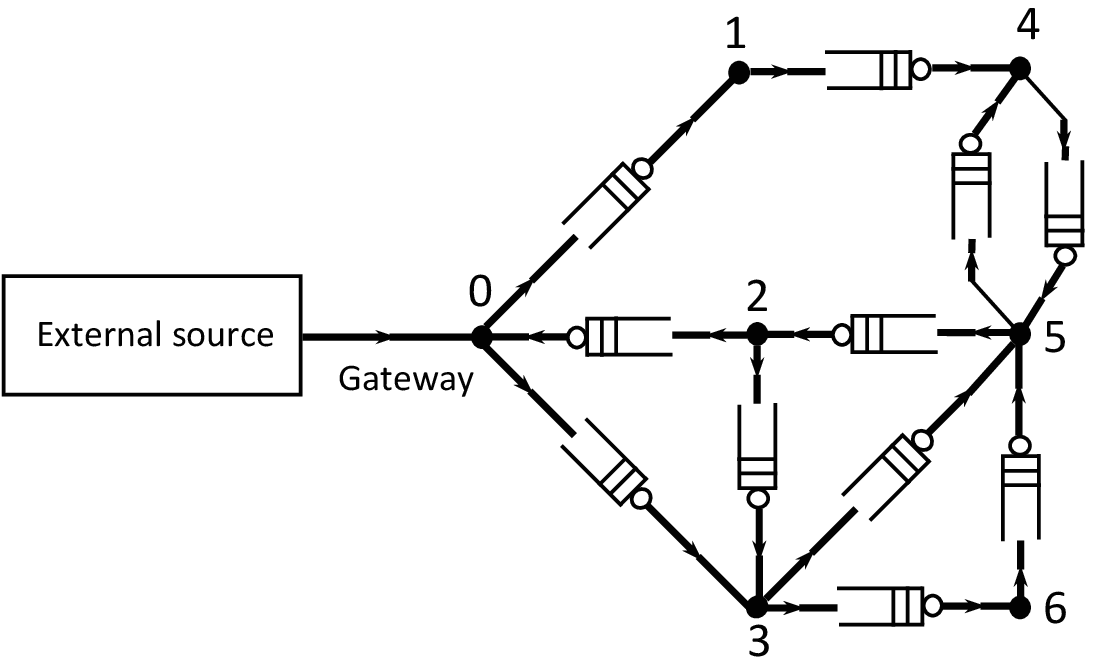}
      \label{model_a}
      }
 \subfigure[Multiple-gateway model]{
  \includegraphics[scale=0.65]{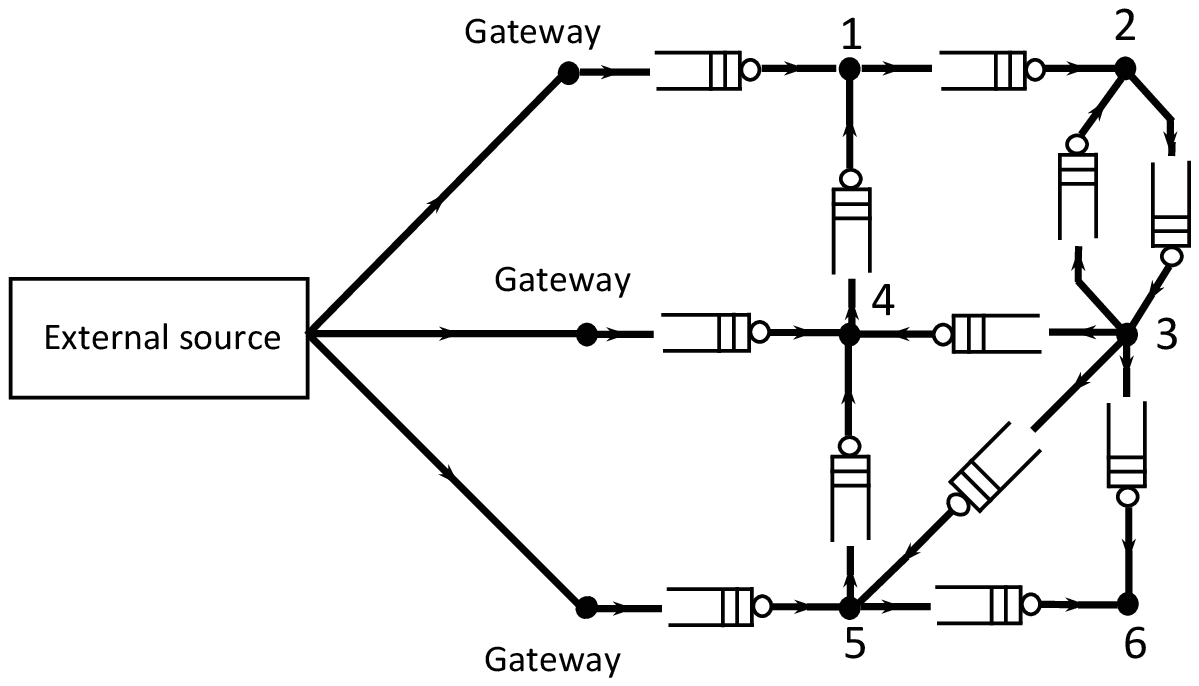}
\label{model_b}   
   }
\captionsetup{justification=justified, font={onehalfspacing}}
\caption{Information updates in multihop networks.}
\end{figure*}\label{Fig:sysMod}

There has been a growing interest in applications that require real-time information updates, such as news, weather reports, email notifications, stock quotes, social updates, mobile ads, etc. The freshness of information is also crucial in other systems, e.g., monitoring systems that obtain information from environmental sensors and wireless communication systems that need rapid updates of channel state information. 

As a metric of data freshness, the \emph{age of information}, or simply \emph{age}, was defined in \cite{adelberg1995applying,cho2000synchronizing,golab2009scheduling,KaulYatesGruteser-Infocom2012}. At time $t$, if the freshest update at the destination was generated at time $U(t)$, the age $\Delta(t)$  is defined as $\Delta(t)=t-U(t)$. Hence, age is the time elapsed since the freshest packet was generated.

The demand for real-time information updates in multihop networks, such as the IoT, intelligent transportation systems, and sensor networks, has gained increasing attention recently. In intelligent transportation systems \cite{Papadimitratos_vehicular,Kaul_vehicul,roy_vehicul}, for example, a vehicle shares its information related to traffic congestion and road conditions to avoid collisions and reduce congestion. Thus, in such applications, maintaining the age at a low level at all network nodes is a crucial requirement. In some other information update applications, such as emergency alerts and sensor networks, critical information is needed to report in a timely manner, and the energy consumption of the sensor nodes must be sufficiently low to support a long battery life up to 10-15 years \cite{timmons2004analysis}. Because of the low traffic load in these systems, wireless interference is not the limiting factor, but rather battery life through energy consumption is. Furthermore, information updates over the Internet, cloud systems, and social networks are of significant importance. These systems are built on wireline networks or implemented based on transport layer APIs. Motivated by these applications, we investigate information updates over multihop networks that can be modeled as multihop queueing systems.

It has been observed in early studies on age of information analysis \cite{CostaCodreanuEphremides2014ISIT,CostaCodreanuEphremides_TIT,Icc2015Pappas,PAoI_in_error,Gamma_dist} that Last-Come, First-Serve (LCFS)-type of scheduling policies can achieve a lower age than other policies. The optimality of the LCFS policy, or more generally the Last-Generated, First-Served (LGFS) policy, for minimizing the age of information in single-hop networks was first established in \cite{age_optimality_multi_server, Bedewy_NBU_journal}. However, age-optimal scheduling in multihop networks remains an important open question.

\begin{table*}[!tbp]
\footnotesize
 \centering
\begin{tabular}{|l|l|l|l|l|l|l|}
\hline
\textbf{Theorem \#} & \textbf{Preemption type} & \textbf{\begin{tabular}[c]{@{}l@{}}Transmission time\\ distribution\end{tabular}} & \textbf{Network topology} & \textbf{Policy space} & \textbf{Proposed policy} & \textbf{Optimality result} \\ \hline
~~~~~~\ref{thm1} & \begin{tabular}[c]{@{}l@{}}Preemption is \\ allowed\end{tabular} & Exponential & General & Causal policies & Preemptive LGFS & Age-optimal \\ \hline
~~~~~~\ref{thmnbu_gab} & \begin{tabular}[c]{@{}l@{}}Preemption is \\ allowed\end{tabular} & \begin{tabular}[c]{@{}l@{}}New-Better-\\ than-Used\end{tabular} & \begin{tabular}[c]{@{}l@{}}Each node has only\\ one incoming link\end{tabular} & Causal policies & \begin{tabular}[c]{@{}l@{}}Non-preemptive \\ LGFS\end{tabular} & Near age-optimal \\ \hline
~~~~~~\ref{thm2}  & \begin{tabular}[c]{@{}l@{}}Preemption is \\ not allowed\end{tabular} & Arbitrary & General & \begin{tabular}[c]{@{}l@{}}Work-conserving\\ causal policies\end{tabular} & \begin{tabular}[c]{@{}l@{}}Non-preemptive \\ LGFS\end{tabular} & Age-optimal \\ \hline
\end{tabular}
\caption{Summary of age-optimality results.}
\label{table1}
\end{table*}


In this paper, we consider a multihop network represented by a directed graph, as shown in Fig. \ref{Fig:sysMod}, where the update packets are generated at an external source and are then dispersed throughout the network via one or multiple gateway nodes. The case of multiple gateway nodes is motivated by news spreading in social media where news is usually posted by multiple social accounts or webpages. Moreover, we suppose that the packet generation times at the external source and the packet arrival times at the gateway node (gateway nodes) are arbitrary. This is because, in some applications, such as sensor and environment monitoring networks, the arrival process is not necessarily Poisson. For example, if a sensor observes an environmental phenomenon and sends an update packet whenever a change occurs, the arrival process of these update packets does not follow a Poisson process. The packet transmission times are independent but not necessarily identically distributed across the links, and
i.i.d. across time. Interestingly, we
find that some low-complexity scheduling policies can achieve (near) age-optimal performance in this setting. The main results in this paper are summarized in Table \ref{table1}.


%

%

\subsection{Our Contributions}
We develop scheduling policies that can achieve age-optimality or near age-optimality in a multihop network with a single information flow. 
The following summarizes our main contributions in this paper: 
\begin{itemize}
\item  If preemption is allowed and the packet transmission times over the network links are exponentially distributed, we prove that the preemptive LGFS policy minimizes the age processes at all nodes in the network among all causal policies in a stochastic ordering sense (Theorem \ref{thm1}). In other words, the preemptive LGFS policy minimizes any \emph{non-decreasing functional of the age processes at all nodes}  in a stochastic ordering sense. Note that the non-decreasing functional of the age processes at all nodes represents a very general class of age metrics in that it includes many age penalty metrics studied in the literature, such as the time-average age \cite{KaulYatesGruteser-Infocom2012,2012ISIT-YatesKaul,CostaCodreanuEphremides2014ISIT,CostaCodreanuEphremides_TIT,Icc2015Pappas,BacinogCeranUysal_Biyikoglu2015ITA,2015ISITYates,RYatesTIT16,Gamma_dist}, average peak age \cite{2015ISITHuangModiano,FCFSGG1,CostaCodreanuEphremides2014ISIT,BacinogCeranUysal_Biyikoglu2015ITA,Gamma_dist,PAoI_in_error},  non-linear age functions \cite{generat_at_will,SunJournal2016}, and age penalty functional at single-hop network \cite{age_optimality_multi_server, Bedewy_NBU_journal}.


\item  Although the preemptive LGFS policy can achieve age-optimality for exponential transmission times, it does not always minimize the age processes for non-exponential transmission times. When preemption is allowed, we investigate an important class of packet transmission time distributions called New-Better-than-Used (NBU) distributions, which are more general than exponential. The network topology we consider here is more restrictive in the sense that each node has one incoming link only. We show that the non-preemptive LGFS policy is within a constant age gap from the optimum average age, and that the gap is independent of the packet generation and arrival times, and buffer sizes (Theorem \ref{thmnbu_gab}).  Our numerical result (Fig. \ref{avg_age3}) shows that the preemptive LGFS policy can be very far from the optimum for non-exponential transmission times, while the non-preemptive LGFS policy is near age-optimal.

\item If preemption is not allowed, then for arbitrary distributions of packet transmission times, we prove that the non-preemptive LGFS policy minimizes the age processes at all nodes among all work-conserving policies in the sense of stochastic ordering (Theorem \ref{thm2}).  Age-optimality here can be achieved even if the transmission time distribution differs from one link to another, i.e., the transmission time distributions are heterogeneous.

\end{itemize}

To the best of our knowledge, these are the first optimal results on minimizing the age of information in multihop queueing networks with arbitrary packet generation and arrival processes.

\section{Related Work}\label{RW}
There exist a number of studies focusing on the analysis of the age and figuring out ways to reduce it in single-hop networks \cite{KaulYatesGruteser-Infocom2012,2012ISIT-YatesKaul,2015ISITHuangModiano,FCFSGG1,CostaCodreanuEphremides2014ISIT,CostaCodreanuEphremides_TIT,Icc2015Pappas,BacinogCeranUysal_Biyikoglu2015ITA,2015ISITYates,RYatesTIT16,PAoI_in_error,Gamma_dist}. In \cite{KaulYatesGruteser-Infocom2012,2012ISIT-YatesKaul}, the update frequency was optimized to minimize the age in First-Come, First-Served (FCFS) queueing systems with exponential service times. It was found that this frequency differs from those that minimize the delay or maximize the throughput. Extending the analysis to multi-class FCFS M/G/1 queue was considered in \cite{2015ISITHuangModiano}.  In \cite{FCFSGG1}, the stationary distributions of the age and peak age in FCFS GI/GI/1 queue was obtained. In \cite{CostaCodreanuEphremides2014ISIT,CostaCodreanuEphremides_TIT,Icc2015Pappas}, it was shown that the age can be reduced by discarding old packets waiting in the queue when a new sample arrives. The age of information under energy replenishment constraints was analyzed in \cite{BacinogCeranUysal_Biyikoglu2015ITA,2015ISITYates}. The time-average age was characterized for multiple sources LCFS information-update systems with and without preemption in  
\cite{RYatesTIT16}. In this study, the authors found that sharing service facility among Poisson sources improves the total age. The work in \cite{PAoI_in_error} analyzed the age in the presence of errors when the service times are exponentially
distributed. Gamma-distributed service times was considered in \cite{Gamma_dist}. The studies
in \cite{PAoI_in_error}, \cite{Gamma_dist} were carried out for LCFS queueing systems with and without preemption. 


It should be noted that in our study, the  packet generation and arrival times are exogenous, i.e., they are not controllable by the scheduler. On the other hand,  the generation times of update packets was optimized for single-hop networks in \cite{BacinogCeranUysal_Biyikoglu2015ITA,2015ISITYates,generat_at_will,SunJournal2016,multi_source_bedewy}. A general class of non-negative, non-decreasing age penalty functions was minimized for single source systems in \cite{generat_at_will,SunJournal2016}. Extending the study to multi-source systems was considered in \cite{multi_source_bedewy}, where sampling and scheduling strategies are jointly optimized to minimize the age. A real-time sampling problem of the Wiener process was solved in \cite{Sun_reportISIT17}: If the sampling times are independent of the observed Wiener process, the optimal sampling problem in \cite{Sun_reportISIT17} reduces to an age of information optimization problem; otherwise, the optimal sampling policy can use knowledge of the Wiener process to achieve better performance than age of information optimization. 

There have also been a few recent studies on the age of information in multihop networks \cite{selen2013age,twohop_energyharvesting,shreedhar2018acp,yates2018age,yates2018age_moments,Atilla_multihop,multihop_madiano,Farazi_multihop_bdcasting}.  The age is analyzed for specific network topologies, e.g., line
or star networks, in \cite{selen2013age}. In \cite{twohop_energyharvesting}, an offline optimal sampling policy was developed to minimize the age in two-hop networks with an energy-harvesting source. A congestion control mechanism that enables timely delivery of the update packets over IP networks was considered in \cite{shreedhar2018acp}. In \cite{yates2018age}, the author analyzed the average age in a multihop line network with Poisson arrival process and exponential service times. This analysis was later extended in \cite{yates2018age_moments} to include age moments and distributions. This paper and \cite{yates2018age,yates2018age_moments} complement each other in the following sense: Our results (i.e., Theorem \ref{thm1}) show that the LCFS policy with preemption in service is age-optimal. However, we do not characterize the achieved optimal age, which was evaluated in \cite{yates2018age,yates2018age_moments}. The authors of \cite{Atilla_multihop} addressed the problem of scheduling in wireless multihop networks with general interference model and multiple flows, assuming that all network queues are adopting an FCFS policy. A similar network model was considered in \cite{multihop_madiano}, where the optimal update policy was obtained for the ``active sources scenario''.  In this scenario, each source can generate a packet at any time, and hence, each source always has a fresh packet to send. The active sources scenario in multihop networks was also considered in \cite{Farazi_multihop_bdcasting}, where nodes take turns broadcasting their updates, and hence each node can act either as a source or a relay. In contrast to our study, the works in \cite{Atilla_multihop,multihop_madiano,Farazi_multihop_bdcasting} considered a time-slotted system, where a packet is transmitted from one node to another in one time slot.

\section{Model and Formulation}\label{sysmod}
\subsection{Notations and Definitions}
For any random variable $Z$ and an event $A$, let $[Z\vert A]$ denote a random variable with the conditional distribution of $Z$ for given $A$, and $\mathbb{E}[Z\vert A]$ denote the conditional expectation of $Z$ for given $A$.

Let $\mathbf{x}=(x_1,x_2,\ldots,x_n)$ and $\mathbf{y}=(y_1,y_2,\ldots,y_n)$ be two vectors in $\mathbb{R}^n$, then we denote $\mathbf{x}\leq\mathbf{y}$ if $x_i\leq y_i$ for $i=1,2,\ldots,n$. A set $U\subseteq \mathbb{R}^n$ is called upper if $\mathbf{y}\in U$ whenever $\mathbf{y}\geq\mathbf{x}$ and $\mathbf{x}\in U$. We will need the following definitions: 
\begin{definition} \textbf{ Univariate Stochastic Ordering:} \cite{shaked2007stochastic} Let $X$ and $Y$ be two random variables. Then, $X$ is said to be stochastically smaller than $Y$ (denoted as $X\leq_{\text{st}}Y$), if
\begin{equation*}
\begin{split}
\mathbb{P}\{X>x\}\leq \mathbb{P}\{Y>x\}, \quad \forall  x\in \mathbb{R}.
 \end{split}
\end{equation*}
\end{definition}
\begin{definition}\label{def_2} \textbf{Multivariate Stochastic Ordering:} \cite{shaked2007stochastic} 
Let $\mathbf{X}$ and $\mathbf{Y}$ be two random vectors. Then, $\mathbf{X}$ is said to be stochastically smaller than $\mathbf{Y}$ (denoted as $\mathbf{X}\leq_\text{st}\mathbf{Y}$), if
\begin{equation*}
\begin{split}
\mathbb{P}\{\mathbf{X}\in U\}\leq \mathbb{P}\{\mathbf{Y}\in U\}, \quad \text{for all upper sets} \quad U\subseteq \mathbb{R}^n.
 \end{split}
\end{equation*}
\end{definition}
\begin{definition} \textbf{ Stochastic Ordering of Stochastic Processes:} \cite{shaked2007stochastic} Let $\{X(t), t\in [0,\infty)\}$ and $\{Y(t), t\in[0,\infty)\}$ be two stochastic processes. Then, $\{X(t), t\in [0,\infty)\}$ is said to be stochastically smaller than $\{Y(t), t\in [0,\infty)\}$ (denoted by $\{X(t), t\in [0,\infty)\}\leq_\text{st}\{Y(t), t\in [0,\infty)\}$), if, for all choices of an integer $n$ and $t_1<t_2<\ldots<t_n$ in $[0,\infty)$, it holds that
\begin{align}\label{law9'}
\!\!\!(X(t_1),X(t_2),\ldots,X(t_n))\!\leq_\text{st}\!(Y(t_1),Y(t_2),\ldots,Y(t_n)),\!\!
\end{align}
where the multivariate stochastic ordering in \eqref{law9'} was defined in Definition \ref{def_2}.
\end{definition}

\subsection{Network Model}
We consider a multihop network represented by a directed graph $\mathcal{G(\mathcal{V},\mathcal{L})}$,  where $\mathcal{V}$ is the set of nodes and $\mathcal{L}$ is the set of links, as shown in Fig. \ref{Fig:sysMod} \footnote{For the simplicity of presentation, we focus on the network model with a single gateway node in the rest of the paper. However, it is not hard to see that our results also hold for networks with multiple gateway nodes.}. The number of nodes in the network is $\vert\mathcal{V}\vert=N$. The nodes are indexed from 0 to $N-1$, where node 0 acts as a gateway node. Define $(i, j)\in\mathcal{L}$ as a link from node $i$ to node $j$, where $i$ is the origin node and $j$ is the destination node. We assume that the links in the network can be active simultaneously, which holds in the applications mentioned in Section \ref{Int}. The packet transmission times are independent 
but not necessarily identically distributed across the links, and \emph{i.i.d.} across time. As will be clear later on, we consider the following transmission time distributions: Exponential distribution, NBU distributions, and arbitrary distribution. In addition, we consider two types of network topology: general network topology and special network topology in which each node has one incoming link. We note that this special network topology is an extension of tandem queues. These different network settings are summarized in Table \ref{table1}.



The system starts to operate at time $t=0$. The update packets are generated at an external source, and are firstly forwarded to node 0, from which they are dispersed throughout the network. Thus, the update packets may arrive at node 0 some time after they
are generated. The $l$-th update packet, called packet $l$, is generated at time $s_l$, arrives at node 0 at time $a_{l0}$, and is delivered to any other node $j$ at time $a_{lj}$ such that  $0\leq s_1\leq s_2\leq \ldots$ and  $s_l \leq a_{l0} \leq a_{lj}$ for all $j = 1, \ldots, N-1$. Note that in this paper, the sequences $\{s_1,s_2,\ldots\}$ and $\{a_{10}, a_{20},\ldots\}$ are \emph{arbitrary}. Hence, the update packets may not arrive at node $0$ in the order of their generation times. For example, packet $l+1$ may arrive at node 0 earlier than packet $l$ such that $s_l \leq s_{l+1}$ but $a_{l0} \geq a_{(l+1)0}$. We suppose that once a packet arrives at node $i$, it is immediately available to all the outgoing links from node $i$. Moreover, the update packets are time-stamped with their generation times such that each node knows the generation times of its received packets. Each link $(i,j)$ has a queue of buffer size $B_{ij}$ to store the incoming packets, which can be infinite, finite, or even zero. If a link has a finite queue buffer size, then the packet that arrives to a full buffer either is dropped or replaces another packet in the queue. 
%
%
%
%
\subsection{Scheduling Policy}\label{Schpolicy} 
 We let $\pi$ denote a scheduling policy that determines the following (at each link): i) Packet assignments to the server, ii) packet preemption if preemption is allowed, iii) packet droppings and replacements when the queue buffer is full. The sequences of packet generation times $\{s_1, s_2, \ldots\}$ and packet arrival times $\{a_{10}, a_{20}, \ldots\}$ at node 0 do not change according to the scheduling policy, while the packet arrival times at other nodes (i.e., $a_{lj}$ for all $l$ and $j=1,\ldots,N-1$) are functions of the scheduling policy $\pi$. 
 We suppose that the packet transmission times over the links are invariant of the scheduling policy and the realization of a packet transmission time at any link is unknown until its
transmission over this link is completed (unless the transmission time is deterministic).

 Let $\Pi$  denote the set of all \emph{causal} policies, in which
scheduling decisions are made based on the history and current information of the system (system information includes the location, arrival times, and generation times of all the packets in the system, and the idle/busy state of all the servers).
  we define several types of policies in $\Pi$: 
  
  A policy is said to be \textbf{preemptive}, if a link can switch to send another packet at any time; the preempted packets can be stored back into the queue if there is enough buffer space and sent out at a later time when the link is available again. In contrast, in a \textbf{non-preemptive} policy, a link must complete sending the current packet before starting to send another packet.
A policy is said to be \textbf{work-conserving}, if each link is busy whenever there are packets waiting in the link's queue.
\subsection{Age Performance Metric}
\begin{figure}
\includegraphics[scale=0.4]{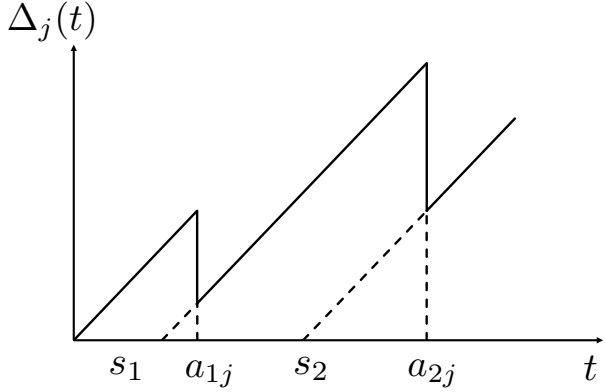}
\centering
\caption{A sample path of the age process $\Delta_j(t)$ at node $j$.}\label{Fig:Age}
\vspace{-0.3cm}
\end{figure}
Let $U_{j}(t)=\max\{s_l : a_{lj}\leq t\}$ be the generation time of the freshest packet arrived at node $j$ before time $t$. The \emph{age of information}, or simply the \emph{age}, at node $j$ is defined as \cite{adelberg1995applying,cho2000synchronizing,golab2009scheduling,KaulYatesGruteser-Infocom2012}
\begin{equation}\label{age}
\begin{split}
\Delta_{j}(t)=t-U_{j}(t).
\end{split}
\end{equation} 
The process of $\Delta_{j}(t)$ is given by $\Delta_{j}=\{\Delta_{j}(t), t\in [0,\infty)\}$. The initial state of $U_j(t)$ at time $t=0^-$ is invariant of the scheduling policy $\pi\in\Pi$, where we assume that $U_j(0^-)=0=s_0$ for all $j\in\mathcal{V}$. As shown in Fig. \ref{Fig:Age}, the age increases linearly with $t$ but is reset to a smaller value with the arrival of a fresher packet. The age vector of all the network nodes at time $t$ is
\begin{align}
\mathbf{\Delta}(t)=\!\!(\Delta_{0}(t), \Delta_{1}(t), \ldots, \Delta_{N-1}(t)). 
\end{align}
The age process of all the network nodes is given by 
\begin{align}
\mathbf{\Delta}=\{\mathbf{\Delta}(t), t\in [0,\infty)\}.
\end{align}

In this paper, we introduce a general \emph{age penalty functional} $g(\mathbf{\Delta})$ to represent the level of dissatisfaction for data staleness at all the network nodes.

\begin{definition}\label{Def_func}  \textbf{Age Penalty Functional:}
Let $\mathbf{V}$ be the set of $n$-dimensional functions, i.e.,
\begin{align}
\mathbf{V} = \{f : [0,\infty)^n \mapsto \mathbb{R}\}.\nonumber
\end{align}
A functional $g:\mathbf{V}\mapsto\mathbb{R}$ is said to be an \emph{age penalty functional} if $g$ is \emph{non-decreasing} in the following sense:  
\begin{equation}
\!\!g(\mathbf{\Delta}_1) \leq g(\bm{\Delta}_2),~\text{whenever}~ \bm \Delta_{1}(t)\leq \bm \Delta_{2}(t), \forall t\in [0,\infty). \!\!
\end{equation}
\end{definition}
The age penalty functionals used in prior studies include:
\begin{itemize}
\item \emph{Time-average age \cite{KaulYatesGruteser-Infocom2012,2012ISIT-YatesKaul,CostaCodreanuEphremides2014ISIT,CostaCodreanuEphremides_TIT,Icc2015Pappas,BacinogCeranUysal_Biyikoglu2015ITA,2015ISITYates,RYatesTIT16,Gamma_dist}:} The time-average age of node $j$ is defined as
\begin{equation}\label{functional1}
g_1(\bm{\Delta})=\frac{1}{T}\int_{0}^{T} \Delta_{j}(t) dt,
\end{equation}
\item \emph{Average peak age \cite{2015ISITHuangModiano,FCFSGG1,CostaCodreanuEphremides2014ISIT,BacinogCeranUysal_Biyikoglu2015ITA,Gamma_dist,PAoI_in_error}:} The average peak age of node $j$ is defined as 
\begin{equation}\label{functional2}
g_2(\bm{\Delta})=\frac{1}{K}\sum_{k=1}^{K} A_{kj},
\end{equation}
where $A_{kj}$ denotes the $k$-th peak value of $\Delta_{j}(t)$ since time $t=0$. 
\item \emph{Non-linear age functions \cite{generat_at_will,SunJournal2016}:} The non-linear age function of node $j$ is in the following form
\begin{equation}\label{functional3}
g_3(\bm{\Delta})= \frac{1}{T}\int_{0}^{T} h(\Delta_{j}(t)) dt,
\end{equation}
where $h$ : $[0,\infty)\to [0,\infty)$ can be any non-negative and non-decreasing function. As pointed out in  \cite{SunJournal2016}, a stair-shape function $h(x)=\lfloor x\rfloor$ can be used to characterize the dissatisfaction of data staleness when the information of interest is checked periodically, and an exponential function $h(x)=e^{x}$ is appropriate for online learning and control applications where the desire for data refreshing grows quickly with respect to the age. Also, an indicator function $h(x)=\mathds{1}(x>d)$ can be used to characterize the dissatisfaction when a given age limit $d$ is violated. 
\item \emph{Age penalty functional in single-hop networks \cite{age_optimality_multi_server, Bedewy_NBU_journal}:} The age penalty functional in \cite{age_optimality_multi_server, Bedewy_NBU_journal} is a non-decreasing functional of the age process at one node, which is a special case of that defined in Definition \ref{Def_func} with $n = 1$.
\end{itemize}

\section{Main Results}\label{GS}
In this section, we present our (near) age-optimality results for multihop networks. We prove our results using stochastic ordering.
\subsection{Exponential Transmission Times, Preemption is Allowed}\label{exp_trans_preemption_on}
 \begin{algorithm}[!t]
\SetKwData{NULL}{NULL}
\SetCommentSty{small} 
$\alpha_{ij}:=0$\tcp*[r]{$\alpha_{ij}$ is the generation time of the packet being transmitted on the link $(i,j)$}
\While{the system is ON} {
\If{a new packet with generation time $s$ arrives to node $i$}{ 
\uIf{the link $(i,j)$ is busy}{
\uIf{ $s\leq\alpha_{ij}$}
{Store the packet in the queue\;}
\Else(~~~~~~~~~~~~~~~~\tcp*[h]{The packet carries fresher information than the packet being transmitted.}){
Send the packet over the link by preempting the packet being transmitted\; 
The preempted packet is stored back to the queue\;
 $\alpha_{ij}=s$\;
}}
\Else(~~~~~~~~~~~~~~~~~~~~~\tcp*[h]{The link is idle.})
{
The new packet is sent over the link\;
} 

}
\If{a packet is delivered to node $j$}{
 \If{ the queue is not empty}{
The freshest packet in the queue is sent over the link\;
 }
}
}
\caption{Preemptive Last-Generated, First-Served policy at the link $(i,j)$.}\label{alg1}
\end{algorithm}

We study age-optimal packet scheduling for networks that allow for preemption and the packet transmission times are exponentially distributed, \emph{independent} across the links and \emph{i.i.d.} across time\footnote{Although we consider exponential transmission times, packet transmission
time distributions are not necessarily identical over the network links, i.e.,
different links may have different mean transmission times.}.  We consider a LGFS scheduling principle which is defined as follows.
\begin{definition}
A scheduling policy is said to follow the \textbf{Last-Generated, First-Served} discipline, if the last generated packet is sent first among all packets in the queue.  
\end{definition}
We consider a preemptive LGFS (prmp-LGFS) policy at each link $(i,j)\in\mathcal{L}$. 
The implementation details of this policy are depicted in Algorithm \ref{alg1}\footnote{The decision related to packet droppings and replacements in full buffer case (at any link) doesn’t affect the age performance of prmp-LGFS policy. Hence, we don’t specify this decision under the prmp-LGFS policy.}. 


Define a set of parameters $\mathcal{I}=\{\mathcal{G}(\mathcal{V}, \mathcal{L}), (B_{ij}, (i,j)\in\mathcal{L}),$  $s_l, a_{l0},l=1,2,\ldots\}$, where $\mathcal{G}(\mathcal{V}, \mathcal{L})$ is the network graph, $B_{ij}$ is the queue buffer size of link $(i,j)$, $s_l$ is the generation time of packet $l$, and $a_{l0}$ is the arrival time of packet $l$ to node $0$. Let $\mathbf{\Delta_{\pi}}$ be the age processes of all nodes in the network under policy $\pi$. The age optimality of prmp-LGFS policy is provided in the following theorem.
\begin{theorem}\label{thm1}
If the packet transmission times are exponentially distributed, \emph{independent} across links and \emph{i.i.d.} across time, then for all $\mathcal{I}$ and $\pi\in\Pi$ 
\begin{align}
[\mathbf{\Delta_{\text{prmp-LGFS}}}\vert\mathcal{I}]\leq_{\text{st}}\!\! [\mathbf{\Delta_{\pi}}\vert\mathcal{I}],
\end{align}
or equivalently, for all $\mathcal{I}$ and non-decreasing functional $g$
 \begin{equation}\label{thm1eq2}
\begin{split}
\mathbb{E}[g(\bm{\Delta}_{\text{prmp-LGFS}})\vert\mathcal{I}]= \min_{\pi\in\Pi} \mathbb{E}[g(\bm{\Delta}_\pi)\vert\mathcal{I}], 
\end{split}
\end{equation}
provided the expectations in \eqref{thm1eq2} exist.
\end{theorem}
\begin{proof}[Proof sketch]
 We use a coupling and forward induction to prove it. We first consider the comparison between the preemptive LGFS policy and any arbitrary policy $\pi$. We couple the packet departure processes at each link of the network such that they are identical under both policies. Then, we use the forward induction over the packet delivery events at each link (using Lemma \ref{lem3}) and the packet arrival events at node 0 (using Lemma \ref{lem4}) to show that the generation times of the freshest packets at each node of the network are maximized under the preemptive LGFS policy. By this, the preemptive LGFS policy is age-optimal among all causal policies. For more details, see Appendix~\ref{Appendix_A}.
\end{proof}


Theorem \ref{thm1} tells us that for arbitrary sequence of packet generation times $\{s_1, s_2, \ldots\}$, sequence of arrival times $\{a_{10}, a_{20}, \ldots\}$ at node $0$, network topology $\mathcal{G(V, L)}$, and buffer sizes $(B_{ij},(i,j)\in\mathcal{L})$, the prmp-LGFS policy achieves optimality of the joint distribution of the age processes at the network nodes within the policy space $\Pi$. In addition, \eqref{thm1eq2} tells us that the prmp-LGFS policy minimizes any non-decreasing age penalty functional $g$, including the time-average age \eqref{functional1}, average peak age \eqref{functional2}, and non-linear age functions \eqref{functional3}. 

As we mentioned before, the result of Theorem \ref{thm1} still holds for the multiple-gateway model shown in Fig. \ref{model_b}. In particular, Lemma \ref{lem4} can be applied to each packet arrival event at each gateway, and hence the result follows. It is also worth pointing out that the arrival processes at the gateway nodes may be heterogeneous, and they do not change according to the scheduling policy.  A weaker version of Theorem \ref{thm1} can be obtained as follows.

\begin{corollary}
If the conditions of Theorem \ref{thm1} hold, then for any arbitrary packet generation and arrival processes at the external source and node 0, respectively, and for all $\pi\in\Pi$
\begin{align}
\mathbf{\Delta_{\text{prmp-LGFS}}}\leq_{\text{st}} \mathbf{\Delta_{\pi}}.
\end{align}
\end{corollary}
\begin{proof}
We consider a mixture over the realizations of packet generation and arrival processes (arrival process at node 0) to prove the result. In particular, by using the result of Theorem \ref{thm1} and Theorem 6.B.16.(e) in \cite{shaked2007stochastic}, the corollary follows.
\end{proof}

\subsection{New-Better-than-Used Transmission Times, Preemption is Allowed}\label{General_transmission_preemption_on}
Although the preemptive LGFS policy can achieve age-optimality when the transmission times are exponentially distributed, it does not always, as we will observe later, minimize the age for non-exponential transmission times. We aim to answer the question
of whether for an important class of distributions that are more general than exponential, optimality or near-optimality can be achieved while preemption is allowed. We here consider the classes of New-Better-than-Used (NBU) packet transmission time distributions, which are defined as follows.
\begin{definition}  \textbf{New-Better-than-Used distributions \cite{shaked2007stochastic}:} Consider a non-negative random variable $X$ with complementary cumulative distribution function (CCDF) $\bar{F}(x)=\mathbb{P}[X>x]$. Then, $X$ is \textbf{New-Better-than-Used (NBU)} if for all $t,\tau\geq0$
\begin{equation}\label{NBU_Inequality}
\bar{F}(\tau +t)\leq \bar{F}(\tau)\bar{F}(t).
\end{equation} 
\end{definition}
Examples of NBU \footnote{The word \textbf{better} in the terminology \textbf{New-Better-than-Used} refers to that a random variable with a long lifetime is better than that with a shorter lifetime \cite{shaked2007stochastic}. In our case, the random variable is the transmission time, and longer transmission time is worse in terms of the age. Thus, the word \textbf{better} here does not imply an improvement in the age performance.} distributions include constant transmission time, (shifted) exponential distribution, geometric distribution, Erlang distribution, negative binomial distribution, etc. Recently,  age was analyzed in single hop networks for exponential transmission times with transmission error in \cite{PAoI_in_error}, and for Gamma-distributed transmission times in \cite{Gamma_dist}. These studies did not answer the question of which policy can be (near) age-optimal for non-exponential transmission times in single hop networks. We provided a unified answer to identify the policy that is near age-optimal in single hop networks in \cite{age_optimality_multi_server, Bedewy_NBU_journal}. Since the question has remained open for multihop networks, we here extend our investigation to answer this question in multihop networks and  identify the near age-optimal policy for a more general class of transmission time distributions.


We propose a non-preemptive LGFS (non-prmp-LGFS) policy. It is important to note that under non-prmp-LGFS policy, the fresh packet replaces the oldest packet in a link's queue when the queue is already at its maximum buffer level (i.e., the queue is already full). 
The implementation details of non-prmp-LGFS policy are depicted in Algorithm \ref{alg2}.
\begin{algorithm}[!t]
\SetKwData{NULL}{NULL}
\SetCommentSty{small} 
$\delta_{ij}:=0$\tcp*[r]{$\delta_{ij}$ is the smallest generation time of the packets in the queue ($B_{ij}$)}
\While{the system is ON} {
\If{a new packet $p_i$ with generation time $s$ arrives to node $i$}{ 
\uIf{the link $(i,j)$ is busy}{
\uIf{ Buffer $(B_{ij})$ is full}
{\uIf{$s>\delta_{ij}$}{
Packet $p_i$ replaces the packet with generation time $\delta_{ij}$ in the queue\;}
\Else{
Drop packet $p_i$\;
}
Set $\delta_{ij}$ to the smallest generation time of the packets in the queue ($B_{ij}$)\;}
\Else{
Store packet $p_i$ in the queue\; 
Set $\delta_{ij}$ to the smallest generation time of the packets in the queue ($B_{ij}$)\;
}}
\Else(~~~~~~~~~~~~~~~~~~~~~\tcp*[h]{The link is idle.})
{
The new packet is sent over the link\;
} 

}
\If{a packet is delivered to node $j$}{
 \If{ the queue is not empty}{
The freshest packet in the queue is sent over the link\;
 }
}
}
\caption{Non-preemptive Last-Generated, First-Served policy at the link $(i,j)$.}\label{alg2}
\end{algorithm}

\begin{figure}
\includegraphics[scale=0.6]{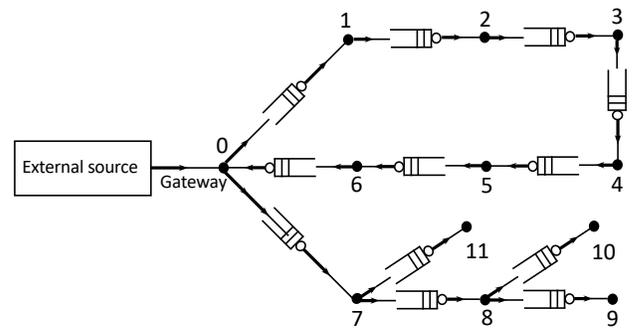}
\centering
\caption{ Information updates over a multihop network, where each node in the network (except the gateway) is restricted to receive data from only one node.}\label{Fig:sysMod2}
\vspace{-0.3cm}
\end{figure}

While we are able to consider a more general class of transmission time distributions, we are able to prove this result for a somewhat more restrictive network than the general topology $\mathcal{G(V,L)}$. The network here is represented by a directed graph $\mathcal{G'(V,L)}$, in which each node $j\in\mathcal{V}\backslash \{0\}$ has one incoming link. An example of this network topology is shown in Fig. \ref{Fig:sysMod2}. We show that the non-prmp-LGFS policy can come close to age-optimal into two steps: i) we construct an infeasible policy which provides the age lower bound, ii) we then show the near age-optimality result by identifying the gap between the constructed lower bound and our proposed policy non-prmp-LGFS. The construction of the the infeasible policy and the lemma that explains the age lower bound are presented in Appendix \ref{Appendix_A'}.


We can now proceed to characterize the age performance of policy non-prmp-LGFS among the policies in $\Pi$. Define a set of the parameters $\mathcal{I}'=\{\mathcal{G'}(\mathcal{V}, \mathcal{L}), (B_{ij}, (i,j)\in\mathcal{L}),$  $s_l, a_{l0},l=1,2,\ldots\}$, where  $\mathcal{G'}(\mathcal{V}, \mathcal{L})$ is the network graph with the new restriction, $B_{ij}$ is the queue buffer size of the link $(i,j)$, $s_l$ is the generation time of packet $l$, and $a_{l0}$ is the arrival time of packet $l$ to node $0$. Define $\mathcal{H}_k$ as the set of nodes in the $k$-th hop, i.e., $\mathcal{H}_k$ is the set of nodes that are separated by $k$ links from node $0$ \footnote{Node $0$ is in $\mathcal{H}_0$.}. Let  $i_{j,k}$ represent the index of the node in $\mathcal{H}_k$ that is in the path to the node $j$ (for example, in Fig. \ref{Fig:sysMod2}, $i_{11,1}=7$ and $i_{10,2}=8$). Define $X_{j}$ as the packet transmission time over the incoming link to node $j$. We use Lemma \ref{lem_lower_bound_all} in Appendix \ref{Appendix_A'} to prove the following theorem.

\begin{theorem}\label{thmnbu_gab}
Suppose that the packet transmission times are NBU, \emph{independent} across links, and \emph{i.i.d.} across time, then for all $\mathcal{I}'$ satisfying $B_{ij}\geq 1$ for each $(i,j)\in\mathcal{L}$
\begin{equation}\label{gap_main1}
\begin{split}
&\min_{\pi\in\Pi}[\bar{\Delta}_{j,\pi}\vert\mathcal{I}']\leq [\bar{\Delta}_{j, \text{non-prmp-LGFS}}\vert\mathcal{I}']\leq\\&
\min_{\pi\in\Pi}[\bar{\Delta}_{j,\pi}\vert\mathcal{I}'] \!+\! \mathbb{E}[X_{i_{j,1}}]\!+\!2\!\sum_{m=2}^k\mathbb{E}[X_{i_{j,m}}],\forall j\in\mathcal{H}_k, \forall k\geq 1,\!\!\!\!\!\!\!\!\!\!\!
\end{split}
\end{equation}
where $\bar{\Delta}_{j,\pi}=\limsup_{T\rightarrow\infty}\frac{\mathbb{E}[\int_{0}^{T}\Delta_{j,\pi}(t)dt]}{T}$ is the average age at node $j$ under policy $\pi$.
\end{theorem}
\begin{proof}[Proof sketch]
We use the infeasible policy and the lower bound process that are constructed in Appendix \ref{Appendix_A'} to prove Theorem \ref{thmnbu_gab} into three steps:

\emph{Step 1:} We derive an upper bound on the time differences between the arrival times (at each node) of the fresh packets under the infeasible policy and those under policy non-prmp-LGFS.

\emph{Step 2:} We use the upper bound derived in Step 1 to derive an upper bound on the average gap between the constructed infeasible policy in Appendix \ref{Appendix_A'} and the non-prmp-LGFS policy.

 \emph{Step 3:} Finally, we use the upper bound on the average gap together with Lemma \ref{lem_lower_bound_all}  in Appendix \ref{Appendix_A'} to prove \eqref{gap_main1}. For the full proof, see Appendix~\ref{Appendix_D}.
\end{proof}
Theorem \ref{thmnbu_gab} tells us that for arbitrary sequence of packet generation times $\{s_1, s_2, \ldots\}$, sequence of arrival times $\{a_{10}, a_{20}, \ldots\}$ at node $0$, and buffer sizes $(B_{ij}\geq 1,(i,j)\in\mathcal{L})$, the non-prmp-LGFS policy is within a constant age gap from the optimum average age among all policies in $\Pi$. Similar to Theorem \ref{thm1}, we can show that the result of Theorem \ref{thmnbu_gab} still holds for the multiple-gateway model shown in Fig. \ref{model_b}.

\begin{remark}
The reason behind considering the restrictive network topology $\mathcal{G'(V, L)}$ is as follows: In the general network topology $\mathcal{G(V,L)}$, a node can receive update packets from multiple paths. As a result, the arrival time of a fresh packet at this node depends on the fastest path that delivers this packet to this node. This fastest path may differ from one packet to another on sample-path. Thus, it becomes challenging to establish an upper bound that is very close to the age lower bound (Steps 1 and 2 in the proof of Theorem \ref{thmnbu_gab}) using sample-path and coupling techniques, in this case.
\end{remark}

\subsection{General Transmission Times, Preemption is Not Allowed}\label{general_no_preemption}
Finally, we study age-optimal packet scheduling for networks that do not allow for preemption and for which the packet transmission times are \emph{arbitrarily} distributed, \emph{independent} across the links and \emph{i.i.d.} across time.  Since preemption is not allowed, we are restricted to non-preemptive policies within $\Pi$. Moreover, we consider work-conserving policies. We use $\Pi_{npwc}\subset\Pi$ to denote the set of non-preemptive work-conserving policies.


%
We consider the non-prmp-LGFS policy, where we show that it is age-optimal among the policies in $\Pi_{npwc}$ in the following theorem.
 
\begin{theorem}\label{thm2}
If the packet transmission times are \emph{independent} across the links and \emph{i.i.d.} across time, then for all $\mathcal{I}$ and $\pi\in\Pi_{npwc}$ 
\begin{align}
\!\!\!\![\mathbf{\Delta_{\text{non-prmp-LGFS}}}\vert\mathcal{I}]
 \!\!\leq_{\text{st}}\!\! [\mathbf{\Delta_{\pi}}\vert\mathcal{I}],\!\!\!
\end{align}
or equivalently, for all $\mathcal{I}$ and non-decreasing functional $g$
 \begin{equation}\label{thm2eq2}
\begin{split}
\mathbb{E}[g(\bm{\Delta}_{\text{non-prmp-LGFS}})\vert\mathcal{I}]= \min_{\pi\in\Pi_{npwc}} \mathbb{E}[g(\bm{\Delta}_\pi)\vert\mathcal{I}], 
\end{split}
\end{equation}
provided the expectations in \eqref{thm2eq2} exist.
\end{theorem}
\begin{proof}
The proof of Theorem \ref{thm2} is similar to that of Theorem \ref{thm1}. The difference is that preemption is not allowed here. See Appendix~\ref{Appendix_B} for more details.
\end{proof}

It is interesting to note from Theorem \ref{thm2} that, age-optimality can be achieved for arbitrary transmission time distributions, even if the transmission time distribution differs from a link to another. General service time distributions have been considered in some recent age analysis on single-hop networks \cite{2015ISITHuangModiano,FCFSGG1}. Theorem \ref{thm2} explains the age-optimal policies in these scenarios. Moreover, similar to Theorem \ref{thm1}, the result of Theorem \ref{thm2} still holds
for the multiple-gateway model shown in Fig. \ref{model_b}.


\begin{remark}
It is worth observing that the results in Theorem \ref{thm1}, Theorem \ref{thmnbu_gab}, and Theorem \ref{thm2} hold for any link buffer sizes $B_{ij}$'s. Hence, the buffer sizes can be chosen according to the application. In particular, in some applications, such as  news and social updates, users are interested in not just the latest updates, but also past news. Thus, in such application, we may need to have queues with buffer sizes greater than one to store old packets and send them later whenever links become idle. On the other hand, there are some other applications, in which old packets become useless when the fresher packets exist. Thus, in these applications, buffer sizes can be chosen to be zero (one) when we follow the prmp-LGFS (non-prmp-LGFS) scheduling policy. 
\end{remark}

\section{Numerical Results}\label{Simulations}
\begin{figure}
\includegraphics[scale=0.6]{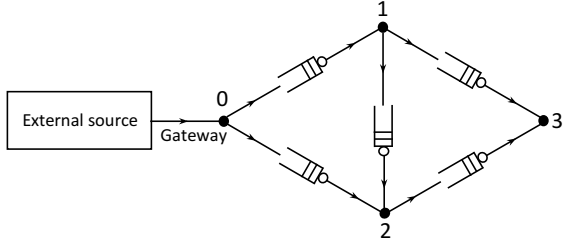}
\centering
\caption{ A  multihop network.}\label{Fig:simulation}
\vspace{-0.3cm}
\end{figure}

We now present numerical results that validate our theoretical findings. The inter-generation times at all setups are \emph{i.i.d.} Erlang-2 distribution with mean $1/\lambda$.


\begin{figure}[t]
\centering
\includegraphics[scale=0.4]{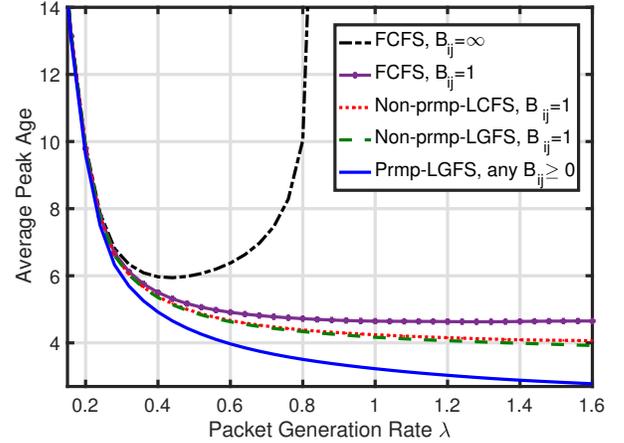}
\caption{Average peak age at node 2 versus packets generation rate $\lambda$ for exponential packet transmission times.}
\vspace{-.0in}
\label{avg_age1}
\end{figure}

We use Figure \ref{avg_age1} to validate the result in Section \ref{exp_trans_preemption_on}. We consider the network in Fig. \ref{Fig:simulation}. The time difference between packet generation and arrival to node $0$, i.e., $a_{i0}-s_i$, is modeled to be either  1 or 100, with equal probability. This means that the update packets may arrive to node $0$ out of order of their generation time. Figure \ref{avg_age1}
 illustrates the average peak age at node 2 versus the packet generation rate $\lambda$ for the multihop network in Fig. \ref{Fig:simulation}. The packet transmission times are exponentially distributed with mean 1 at links $(0,1)$ and $(1,2)$, and mean 0.5 at link $(0,2)$.
Note that the age performance of the preemptive LGFS policy is not affected by the buffer sizes. This is because, in the case of the preemptive LGFS policy, queues are only used to store the old packets, while a fresh packet can start service as soon as it arrives at a queue. Hence, the preemptive LGFS policy has the same performance for different buffer sizes.
One can observe that the preemptive LGFS policy achieves a better (smaller) peak age at node $2$ than the non-preemptive LGFS policy, non-preemptive LCFS policy, and FCFS policy, where the buffer sizes are either 1 or infinity. It is important to emphasize that the peak age is minimized by preemptive LGFS policy for out of order packet receptions at node $0$, and general network topology. This numerical result shows agreement with Theorem \ref{thm1}.

\begin{figure}[t]
\centering
\includegraphics[scale=0.4]{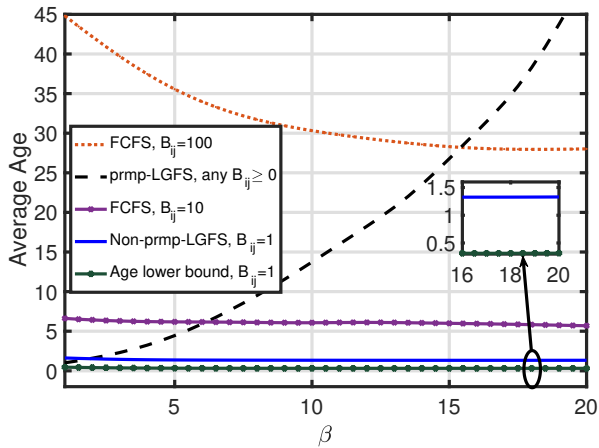}
\caption{Average age at node 5 under gamma transmission time distributions at each link with different shape parameter $\beta$.}
\vspace{-.0in}
\label{avg_age3}
\end{figure}

We use Figure \ref{avg_age3} to validate the results in Section \ref{General_transmission_preemption_on}. We consider the network in Fig. \ref{Fig:sysMod2}.
Figure \ref{avg_age3} illustrates the average age at node 5 under gamma transmission time distributions at each link with different shape parameter $\beta$, where the buffer sizes are either 1, 10, or 100. The mean of the gamma transmission time distributions at each link is normalized to 0.2. The time difference ($a_{i0}-s_i$) between packet generation and arrival to node $0$ is Zero. Note that the average age of the FCFS policy with infinite buffer sizes is extremely high in this case and hence is not plotted in this figure. The ``Age lower bound'' curve is generated by using $\frac{\int_{0}^{T}\Delta_{5,IP}^{LB}}{T}$ when the buffer sizes are 1 which, according to Lemma \ref{lem_lower_bound_all}, is a lower bound of the optimum average age at node 5. We can observe that the gap between the ``Age lower bound'' curve and the average age of the non-prmp-LGFS policy at node 5 is no larger than $9E[X] = 1.8$, which agrees with Theorem \ref{thmnbu_gab}. In addition, we can observe that prmp-LGFS policy achieves the best age performance among all
plotted policies when $\beta=1$. This is because a gamma distribution with shape parameter $\beta=1$ is an exponential distribution. Thus, age-optimality can be achieved in this case by policy prmp-LGFS as stated in Theorem \ref{thm1}. However, as can be seen in the figure, the average age at node 5 of the prmp-LGFS policy blows up as the shape parameter $\beta$ increases and the non-prmp-LGFS policy achieves the best age performance among all plotted policies when $\beta>2$. The reason for this phenomenon is as follows: As $\beta$ increases, the variance (variability) of normalized gamma distribution decreases. Hence, when a packet is preempted, the service time of a new packet is probably longer than the remaining service time of the preempted packet. Because the generation rate is high, packet preemption happens frequently, which leads to infrequent packet delivery and increases the age. This phenomenon occurs heavily at the first link (link $(0,1)$) which, in turn, affects the age at the subsequent nodes.

\begin{figure}[t]
\centering
\includegraphics[scale=0.4]{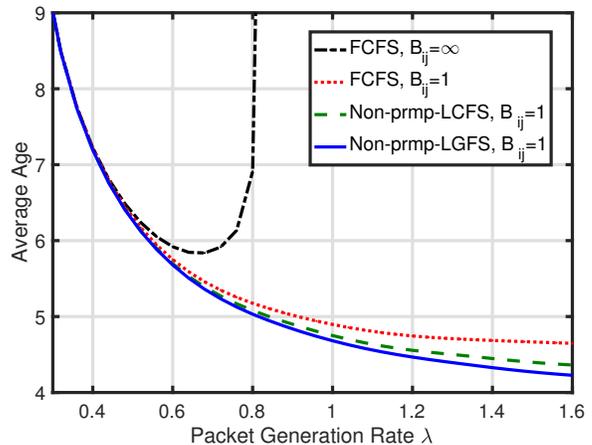}
\caption{Average age at node 3 versus packets generation rate $\lambda$ for general packet transmission time distributions.}
\vspace{-.0in}
\label{avg_age2}
\end{figure}
We use Figure \ref{avg_age2} to validate the result in Section \ref{general_no_preemption}.
We consider the network in Fig. \ref{Fig:simulation}. The time difference between packet generation and arrival to node $0$, i.e., $a_{i0}-s_i$, is modeled to be either  1 or 100, with equal probability. Figure \ref{avg_age2} plots  the time-average age at node 3 versus the packets generation rate $\lambda$ for the multihop network in Fig. \ref{Fig:simulation}. The plotted policies are FCFS policy, non-preemptive LCFS, and non-preemptive LGFS policy, where the buffer sizes are either 1 or infinity. The packet transmission times at links $(0,1)$ and $(1,3)$ follow a gamma distribution with mean 1. The packet transmission times at links $(0, 2)$, $(1, 2)$, and $(2, 3)$ are distributed as the sum of a constant with value 0.5 and a value drawn from an exponential distribution with mean 0.5. We find that the non-preemptive LGFS policy achieves the best age performance among all plotted policies. By comparing the age performance of the non-preemptive LGFS  and non-preemptive LCFS policies, we observe that the LGFS scheduling principle improves the age performance when the update packets arrive to node 0 out of the order of their generation times. It is important to note that the non-preemptive LGFS policy minimizes the age among the non-preemptive work-conserving policies even if the packet transmission time distributions are heterogeneous across the links. This observation agrees with Theorem \ref{thm2}. We also observe that the average age of FCFS policy with $B_{ij}=\infty$ blows up when the traffic intensity is high. This is due to the increased congestion in the network which leads to a delivery of stale packets. Moreover, in case of the FCFS policy with $B_{ij}=1$, the average age is finite at high traffic intensity, since the fresh packet has a better opportunity to be delivered in a relatively short period compared with FCFS policy with $B_{ij}=\infty$. 

\section{Conclusion}\label{Concl}
In this paper, we studied the age minimization problem in interference-free multihop networks. We considered general system settings including arbitrary network topology, packet generation and arrival times at node 0, and queue buffer sizes. A number of scheduling policies were developed and proven to be (near) age-optimal in  a  stochastic  ordering  sense for minimizing any non-decreasing functional of the age processes. In particular, we showed that age-optimality can be achieved when: i) preemption is allowed and the packet transmission times are exponentially distributed, ii) preemption is not allowed and the packet transmission times are arbitrarily distributed (among work-conserving policies). Moreover, for networks that allow for preemption and the packet transmission times are NBU, we showed that the non-preemptive LGFS policy is near age-optimal in a somewhat more restrictive network topology.

\appendices
\section{Proof of Theorem \ref{thm1}}\label{Appendix_A}
Let us define the system state of a policy $\pi$:

 \definition  At any time $t$, the \emph{system state} of policy $\pi$ is specified by  $\mathbf{U}_\pi(t)=( U_{0,\pi}(t), U_{2,\pi}(t),$ $\ldots,U_{N-1,\pi}(t)) $, where $U_{j,\pi}(t)$ is the generation time of the freshest packet that arrived at node $j$ by time $t$.
Let $\{\mathbf{U}_\pi(t), t\in[0,\infty)\}$ be the state process of policy $\pi$, which is assumed to be right-continuous. For notational simplicity, let policy $P$ represent the preemptive LGFS policy.

The key step in the proof of Theorem \ref{thm1} is the following lemma, where we compare policy $P$ with any work-conserving policy $\pi$.

 \begin{lemma}\label{lem2}
 Suppose that $\mathbf{U}_{P}(0^-)=\mathbf{U}_{\pi}(0^-)$ for all work conserving policies $\pi$, then for all $\mathcal{I}$,
\begin{equation}\label{law9}
\begin{split}
[\{\mathbf{U}_{P}(t),  t\in[0,\infty)\}\vert\mathcal{I}]\geq_{\text{st}}[\{\mathbf{U}_{\pi}(t), t\in[0,\infty)\}\vert\mathcal{I}].
 \end{split}
\end{equation}
\end{lemma}

 We use coupling and forward induction to prove Lemma \ref{lem2}.
For any work-conserving policy $\pi$, suppose that stochastic processes $\widetilde{\mathbf{U}}_{P}(t)$ and $\widetilde{\mathbf{U}}_{\pi}(t)$ have the same distributions with $\mathbf{U}_{P}(t)$ and $\mathbf{U}_{\pi}(t)$, respectively. 
The state processes $\widetilde{\mathbf{U}}_{{P}}(t)$ and $\widetilde{\mathbf{U}}_{\pi}(t)$ are coupled in the following manner: If a packet is delivered from node $i$ to node $j$ at time $t$ as $\widetilde{\mathbf{U}}_{{P}}(t)$ evolves in policy $P$,  then there exists a packet delivery from node $i$ to node $j$ at time $t$ as $\widetilde{\mathbf{U}}_{\pi}(t)$ evolves in policy $\pi$. 
Such a coupling is valid since the transmission time is exponentially distributed and thus memoryless. Moreover, policy ${P}$ and policy $\pi$ have identical packet generation times $(s_1, s_2, \ldots, s_n)$ at the external source and packet arrival times $(a_{10}, a_{20}, \ldots, a_{n0})$ to node 0. According to Theorem 6.B.30 in \cite{shaked2007stochastic}, if we can show 
\begin{equation}\label{main_eq}
\begin{split}
\mathbb{P}[\widetilde{\mathbf{U}}_{P}(t)\geq\widetilde{\mathbf{U}}_{\pi}(t), t\in[0,\infty)\vert\mathcal{I}]=1,
\end{split}
\end{equation}
then \eqref{law9} is proven. 

To ease the notational burden, we will omit the tildes in this proof on the coupled versions and just use $\mathbf{U}_{P}(t)$ and $\mathbf{U}_{\pi}(t)$. Next, we use the following lemmas to prove \eqref{main_eq}:

\begin{lemma}\label{lem3}
Suppose that under policy $P$, $\mathbf{U'}_{P}$ is obtained by a packet delivery over the link $(i,j)$ in the system whose state is $\mathbf{U}_{P}$. Further, suppose that under policy $\pi$, $\mathbf{U'}_{\pi}$ is obtained by a packet delivery over the link $(i,j)$ in the system whose state is $\mathbf{U}_\pi$. If
\begin{equation}\label{hyp1}
 \mathbf{U}_{P} \geq \mathbf{U}_\pi,
\end{equation}
then,
\begin{equation}\label{law6}
\mathbf{U'}_{P} \geq \mathbf{U'}_{\pi}.
\end{equation}
\end{lemma}

\begin{proof}\ifreport
Let $s_{P}$ and $s_\pi$ denote the generation times of the packets that are delivered over the link $(i,j)$ under policy $P$ and policy $\pi$, respectively. From the definition of the system state, we can deduce that
\begin{equation}\label{Def1}
\begin{split}
U_{j,P}'&=\max\{U_{j,P},s_{P}\},\\
U_{j,\pi}'&=\max\{U_{j,\pi},s_{\pi}\}.
\end{split}
\end{equation}
Hence, we have two cases:

Case 1: If $s_{P}\geq s_\pi$. From \eqref{hyp1}, we have
\begin{equation}\label{pfp21}
U_{j,P}\geq U_{j,\pi}.
\end{equation}
Also, $s_{P}\geq s_\pi$, together with  \eqref{Def1} and \eqref{pfp21} imply
\begin{equation}
U'_{j,P}\geq U'_{j,\pi}.
\end{equation}
Since there is no packet delivery under other links, we get
\begin{equation}
U'_{k,P}=U_{k,P}\geq U_{k,\pi}=U'_{k,\pi}, \quad \forall k\neq j.
\end{equation}
Hence, we have 
\begin{equation}
\mathbf{U'}_{P} \geq \mathbf{U'}_{\pi}.
\end{equation}

Case 2: If $s_{P}<s_\pi$. By the definition of the system state, $s_{P} \leq U_{i,P}$ and $s_\pi \leq U_{i,\pi}$. Then, using $U_{i, P}\geq U_{i,\pi}$, we obtain
\begin{equation}
s_{P} < s_\pi \leq U_{i,\pi} \leq U_{i, P}.
\end{equation}
Because $s_{P}<U_{i, P}$, policy $P$ is sending a stale packet on link $(i,j)$. By the definition of policy $P$, this happens only when all packets that are generated after $s_P$ in the queue of the link $(i,j)$ have been delivered to node $j$. Since $s_\pi \leq U_{i, P}$, node $i$ has already received a packet (say packet $w$) generated no earlier than $s_\pi$ in policy $P$. Because $s_{P}<s_\pi$, packet $w$ is generated after $s_{P}$. Hence, packet $w$ must have been delivered to node $j$ in policy $P$ such that
\begin{equation}\label{eqprmp1}
 s_\pi  \leq U_{j, P}.
\end{equation}
Also, from \eqref{hyp1}, we have
\begin{equation}\label{eqprmp2}
 U_{j,\pi} \leq U_{j, P}.
\end{equation}
Combining \eqref{eqprmp1} and \eqref{eqprmp2} with \eqref{Def1}, we obtain
\begin{equation}
U'_{j,P}\geq U'_{j,\pi}.
\end{equation}
Since there is no packet delivery under other links, we get
\begin{equation}
U'_{k,P}=U_{k,P}\geq U_{k,\pi}=U'_{k,\pi}, \quad \forall k\neq j.
\end{equation}
Hence, we have 
\begin{equation}
\mathbf{U'}_{P} \geq \mathbf{U'}_{\pi},
\end{equation}
which complete the proof. 
\else
See our technical report \cite{Technical_report}.\fi
\end{proof}

\begin{lemma}\label{lem4}
Suppose that under policy $P$, $\mathbf{U'}_{P}$ is obtained by the arrival of a new packet to node $0$ in the system whose state is $\mathbf{U}_{P}$. Further, suppose that under policy $\pi$, $\mathbf{U'}_{\pi}$ is obtained by the arrival of a new packet to node $0$ in the system whose state is $\mathbf{U}_\pi$. If
\begin{equation}\label{hyp2}
 \mathbf{U}_{P} \geq \mathbf{U}_\pi,
\end{equation}
then,
\begin{equation}
\mathbf{U'}_{P} \geq \mathbf{U'}_{\pi}.
\end{equation}
\end{lemma}

\begin{proof}\ifreport
Let $s$ denote the generation time of the new arrived packet. From the definition of the system state, we can deduce that
\begin{equation}\label{Def2}
\begin{split}
U_{0,P}'&=\max\{U_{0,P},s\},\\
U_{0,\pi}'&=\max\{U_{0,\pi},s\}.
\end{split}
\end{equation}
Combining this with \eqref{hyp2}, we obtain
\begin{equation}
U'_{0,P}\geq U'_{0,\pi}.
\end{equation}
Since there is no packet delivery under other links, we get
\begin{equation}
U'_{k,P}=U_{k,P}\geq U_{k,\pi}=U'_{k,\pi}, \quad \forall k\neq 0.
\end{equation}
Hence, we have 
\begin{equation}
\mathbf{U'}_{P} \geq \mathbf{U'}_{\pi},
\end{equation}
which complete the proof.
\else
See our technical report \cite{Technical_report}.\fi 
\end{proof}

\begin{proof}[Proof of Lemma \ref{lem2}]
For any sample path, we have that $\mathbf{U}_{P}(0^-) = \mathbf{U}_{\pi}(0^-)$. This, together with Lemma \ref{lem3} and Lemma \ref{lem4},  implies that  
\begin{equation}
\begin{split}
[\mathbf{U}_{P}(t)\vert\mathcal{I}] \geq [\mathbf{U}_{\pi}(t)\vert\mathcal{I}],\nonumber
\end{split}
\end{equation}
holds for all $t\in[0,\infty)$. Hence, \eqref{main_eq} holds which implies \eqref{law9} by Theorem 6.B.30 in \cite{shaked2007stochastic}.
This completes the proof.
\end{proof}

\begin{proof}[Proof of Theorem \ref{thm1}]
According to Lemma \ref{lem2}, we have
\begin{equation*}
\begin{split}
[\{\mathbf{U}_{P}(t),  t\in[0,\infty)\}\vert\mathcal{I}]\geq_{\text{st}} [\{\mathbf{U}_{\pi}(t), t\in[0,\infty)\}\vert\mathcal{I}],
 \end{split}
\end{equation*}
holds for all work-conserving policies $\pi$, which implies
\begin{equation*}
\begin{split}
[\{\mathbf{\Delta}_{P}(t), t\in[0,\infty)\}\vert\mathcal{I}]\!\!\leq_{\text{st}} \!\![\{\mathbf{\Delta}_{\pi}(t), t\in[0,\infty)\}\vert\mathcal{I}],
 \end{split}
\end{equation*}
holds for all work-conserving policies $\pi$.

Finally, transmission idling only postpones the delivery of fresh packets. Therefore, the age under non-work-conserving policies will be greater. As a result,
\begin{equation*}
\begin{split}
[\{\mathbf{\Delta}_{P}(t), t\in[0,\infty)\}\vert\mathcal{I}]\!\!\leq_{\text{st}} \!\![\{\mathbf{\Delta}_{\pi}(t), t\in[0,\infty)\}\vert\mathcal{I}],
 \end{split}
\end{equation*}
holds for all $\pi\in\Pi$. This completes the proof.
\end{proof}

\section{Lower bound construction}\label{Appendix_A'}
Let $v_{lj}(\pi)$ denote the transmission starting time of packet $l$ over the incoming link to node $j$ under policy $\pi$. We construct an infeasible policy which provides the age lower bound as follows:

\begin{itemize}
\item[1-] The infeasible policy ($IP$) is constructed as follows. At each link $(i,j)$, the packets are served by following a work-conserving LGFS principle. A packet $l$ is deemed delivered from node $i$ to node $j$ once the transmission of packet $l$ starts over the link $(i,j)$ (this step is infeasible). After the transmission of packet $l$ starts over the link $(i,j)$, the link $(i,j)$ will be busy for a time duration equal to the actual transmission time of packet $l$ over the link $(i,j)$. Hence, the next packet cannot start its transmission over the link $(i,j)$ until the end of this time duration. We use $v_{lj}(IP)$ to denote the transmission starting time of packet $l$ over the incoming link to node $j$ under the infeasible policy ($IP$) constructed above.

 \begin{figure}[!tbp]
 \centering
 \subfigure[Two-hop tandem network.]{
  \includegraphics[scale=0.9]{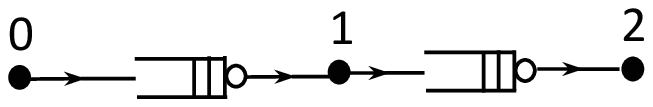}
   \label{f_a}
   }
 \subfigure[A sample path of the packet arrival processes.]{
  \includegraphics[scale=0.36]{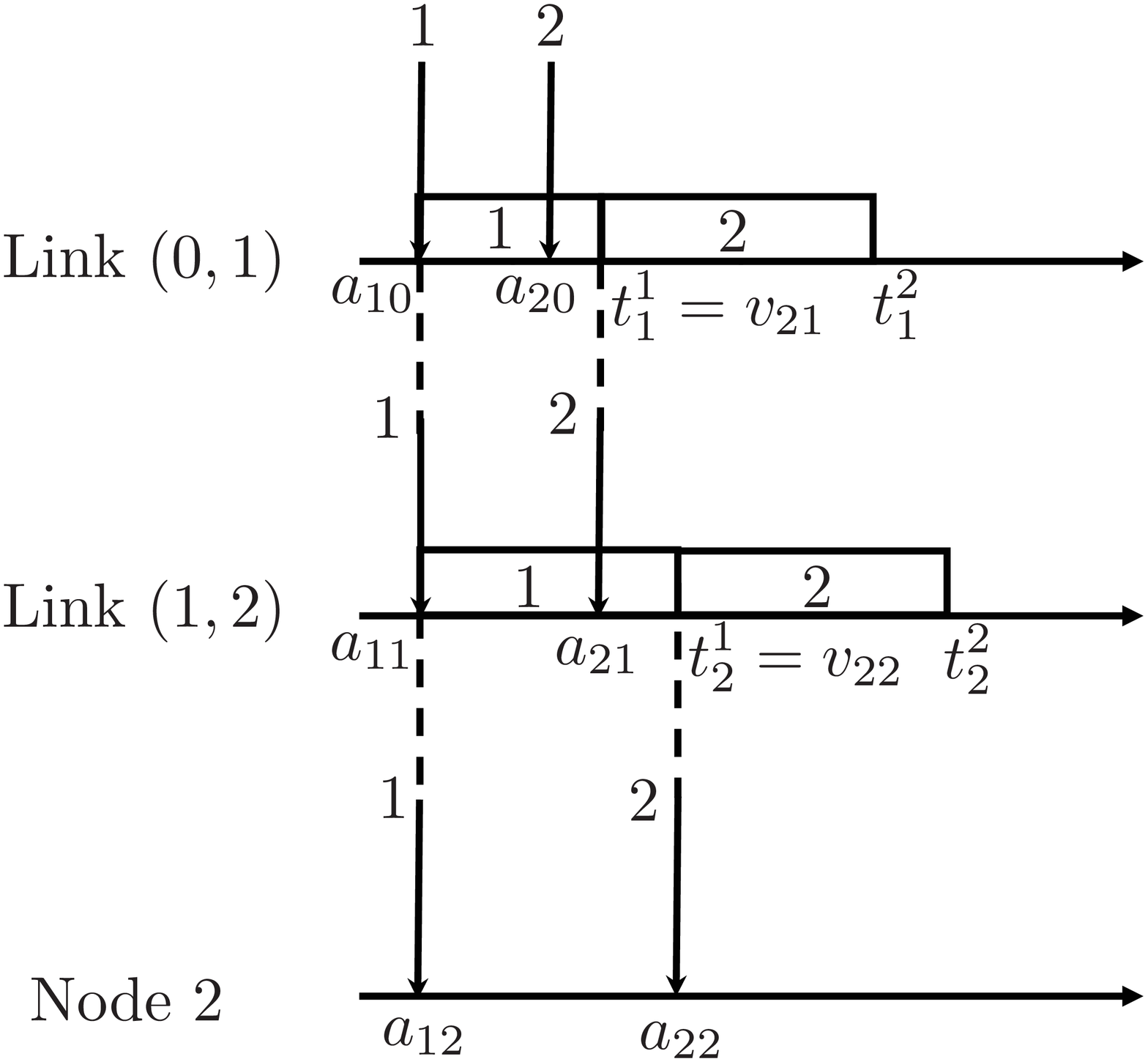}
   \label{f_b}
   }
\caption{An illustration of the infeasible policy in a Two-hop network.} \label{fig:arrival_sample_1}
\end{figure}
One example of the infeasible policy IP is illustrated in Fig. \ref{fig:arrival_sample_1}, where we consider two hops of tandem queues. We use $t^l_j$ to denote the time by which the incoming link to node $j$ becomes idle again after the transmission of packet $l$ starts.  Since all links are idle at the beginning, packet 1 arrives to all nodes once it arrives to node 0 at time $a_{10}$ (this is because each packet is deemed delivered to the next node once its transmission starts). However, each link is kept busy for a time duration equal to the actual transmission time of packet 1 over each link. Then, packet 2 arrives to node 0 at time $a_{20}$ and finds the link $(0,1)$ busy. Therefore, packet 2 cannot start its transmission until link $(0,1)$ becomes idle again at time $t^1_1$ $(v_{21}(IP)=t^1_1)$. Once packet 2 starts its transmission at time $t^1_1$ over the link $(0,1)$, it is deemed delivered to node 1 $(a_{21}(IP)=v_{21}(IP)=t^1_1)$ and link $(0,1)$ is kept busy until time $t^2_1$. At time $t^1_1$, link $(1,2)$ is busy. Thus, packet 2 cannot start its transmission over the link $(1,2)$ until it becomes idle again at time $t^1_2$. Once packet 2 starts its transmission over the link $(1,2)$ at time $t^1_2$, it is deemed delivered to node 2 $(a_{22}(IP)=v_{22}(IP)=t^1_2)$.

\item[2-] The age lower bound is constructed as follows. For each node $j\in\mathcal{V}$, define a function $\Delta_{j,IP}^{\text{LB}}(t)$ as
\begin{equation}\label{lower_bound_def_1}
\Delta_{j,IP}^{\text{LB}}(t)=t-\max\{s_l : v_{lj}(IP)\leq t\}.
\end{equation}
The definition of the $\Delta_{j,IP}^{\text{LB}}(t)$ is similar to that of the age in \eqref{age} except that the packets arrival times to node $j$ are replaced by their transmission starting times over the incoming link to node $j$ in the infeasible policy. In this case, $\Delta_{j,IP}^{\text{LB}}(t)$ increases linearly with $t$ but is reset to a smaller value with the transmission start of a fresher packet over the incoming link to node $j$, as shown in Fig. \ref{Fig:age_lower_at_node_one}. The process of $\Delta_{j,IP}^{\text{LB}}(t)$ is given by $\Delta_{j,IP}^{\text{LB}}=\{\Delta_{j,IP}^{\text{LB}}(t), t\in [0,\infty)\}$ for each $j\in\mathcal{V}$. The age lower bound vector of all the network nodes is 
\begin{align}
\mathbf{\Delta}^{\text{LB}}_{IP}(t)=\!\!(\Delta_{0,IP}^{\text{LB}}(t), \Delta_{1,IP}^{\text{LB}}(t), \ldots, \Delta_{N-1,IP}^{\text{LB}}(t)). 
\end{align}
The age lower bound process of all the network nodes is given by
\begin{align}
\mathbf{\Delta}_{IP}^{\text{LB}}=\{\mathbf{\Delta}^{\text{LB}}_{IP}(t), t\in [0,\infty)\}.
\end{align}
\begin{figure}
\centering \includegraphics[scale=0.38]{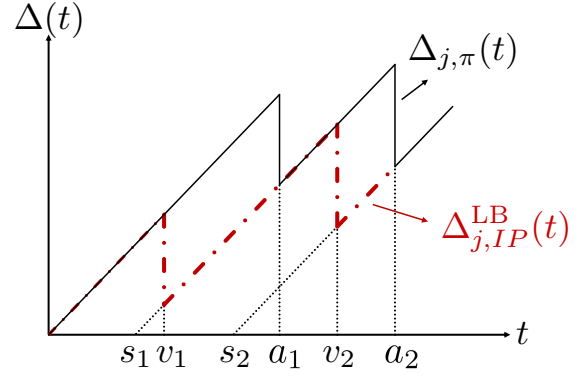}
\centering
\captionsetup{justification=justified, font={onehalfspacing}}
\caption{The evolution of $\Delta_{j,IP}^{\text{LB}}(t)$ and $\Delta_{j,\pi}(t)$ at node $j\in\mathcal{H}_1$. For figure clarity, we use $v_1$ and $v_2$ to denote $v_{1j}(IP)$ and $v_{2j}(IP)$, respectively. Also, we use $a_1$ and $a_2$ to denote $a_{1j}(\pi)$ and $a_{2j}(\pi)$, respectively. We suppose that $a_{10}>s_1$ and $a_{20}>a_1>s_2$, such that $a_{10}=v_1$ and $a_{20}=v_2$. }\label{Fig:age_lower_at_node_one}
\vspace{-0.3cm}
\end{figure}
\end{itemize}

The next Lemma tells us that the process $\mathbf{\Delta}_{IP}^{\text{LB}}$ is an age lower bound of all policies in $\Pi$ in the following sense. 
\begin{lemma}\label{lem_lower_bound_all}
Suppose that the packet transmission times are NBU, \emph{independent} across links, and \emph{i.i.d.} across time, then for all $\mathcal{I}'$ satisfying $B_{ij}\geq 1$ for each $(i,j)\in\mathcal{L}$, and $\pi\in\Pi$
\begin{align}\label{lower_bound_equation}
[\mathbf{\Delta}^{\text{LB}}_{IP}\vert\mathcal{I}']\leq_{\text{st}} [\mathbf{\Delta_{\pi}}\vert\mathcal{I}'].
\end{align}
\end{lemma}
\begin{proof}
Condition \eqref{NBU_Inequality} is very crucial in proving Lemma \ref{lem_lower_bound_all}.  In particular,  \eqref{NBU_Inequality} implies that for NBU service time distributions, the remaining service time of a packet that has already spent $\tau$ seconds in service is probably shorter than the service time of a new packet (i.e., $\mathbb{P}[X>t+\tau\vert X>\tau]\leq \mathbb{P}[X>t]$, where $X$ represents the service time). This is used to show that the transmission starting times of the fresh packets under policy $IP$ are stochastically smaller than their corresponding delivery times under policy $\pi$, and hence \eqref{lower_bound_equation} follows. For more details, see Appendix~\ref{Appendix_C}.
\end{proof}

\section{Proof of Lemma \ref{lem_lower_bound_all}}\label{Appendix_C}

For notation simplicity, let policy $IP$ represent the infeasible policy. We need to define the following parameters: Recall that $v_{lj}$ denotes the transmission starting time of packet $l$ over the incoming link to node $j$ and $a_{lj}$ denotes the arrival time of packet $l$ to node $j$. We define $\Gamma_{lj}$ and $D_{lj}$ as
\begin{equation}\label{Def_na1_1}
\begin{split}
\Gamma_{lj}&=\min_{q\geq l}\{v_{qj}\},
\end{split}
\end{equation}
\begin{equation}\label{Def_na1_2}
D_{lj}=\min_{q\geq l}\{a_{qj}\},
\end{equation}
\begin{figure}
\centering
\includegraphics[scale=0.27]{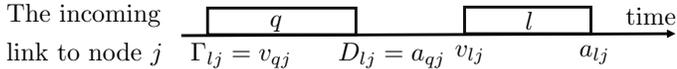}
\centering
\captionsetup{justification=justified, font={small,onehalfspacing}}
\caption{ An illustration of $v_{lj}$, $a_{lj}$, $\Gamma_{lj}$ and $D_{lj}$. We consider the incoming link to node $j$, and $s_q>s_l$. The transmission starting time over this link and the arrival time to node $j$ of packet $q$ are earlier than those of packet $l$. Thus, we have $\Gamma_{lj}=v_{qj}$ and $D_{lj}=a_{qj}$.}\label{Fig:parameter}
\vspace{-0.3cm}
\end{figure}
where $\Gamma_{lj}$  and $D_{lj}$ are the smallest transmission starting time over the incoming link to node $j$ and arrival time to node $j$, respectively, of all packets that are fresher than the packet $l$. An illustration of these parameters is provided in Fig. \ref{Fig:parameter}. Suppose that there are $n$ update packets, where $n$ is an arbitrary positive integer, no matter finite or infinite. Define the vectors $\mathbf{\Gamma}_j=(\Gamma_{1j}, \ldots, \Gamma_{nj})$, and $\mathbf{D}_j=(D_{1j}, \ldots, D_{nj})$. Also, a packet $l$ is said to be an informative packet at node $i$, if all packets that arrive to node $i$ before packet $l$ are staler than packet $l$, i.e., $s_{l'}\leq s_{l}$ for all packets $l'$ satisfying $a_{l'i}\leq a_{li}$. All these quantities are functions of the scheduling policy $\pi$ (except the packet arrival times $(a_{10}, a_{20}, \ldots, a_{n0})$ to node 0 which are invariant of the scheduling policy). 

We can deduce from \eqref{age} that the age process $\mathbf{\Delta_{\pi}}$ under any policy $\pi$ is  increasing in $[\mathbf{D}_1(\pi), \ldots, \mathbf{D}_{N-1}(\pi)]$. Moreover, we can deduce from \eqref{lower_bound_def_1} that the process $\mathbf{\Delta}_{IP}^{\text{LB}}$ is increasing in $[\mathbf{\Gamma}_1(IP), \ldots, \mathbf{\Gamma}_{N-1}(IP)]$. According to Theorem 6.B.16.(a) of \cite{shaked2007stochastic}, if we can show 
\begin{equation}\label{Lower_bound_eq1}
[\mathbf{\Gamma}_1(IP), \ldots, \mathbf{\Gamma}_{N-1}(IP)\vert\mathcal{I}']\leq_{\text{st}} [\mathbf{D}_1(\pi), \ldots, \mathbf{D}_{N-1}(\pi)\vert\mathcal{I}'],
\end{equation}
holds for all $\pi\in\Pi$, then \eqref{lower_bound_equation} is proven. Hence, \eqref{Lower_bound_eq1} is what we need to show. We pick an arbitrary policy $\pi\in\Pi$ and prove \eqref{Lower_bound_eq1} into two steps:

\emph{Step 1}: We first show that, at any link $(i,j)$, if the arrival times of the informative packets at node $i$ under policy $IP$ are earlier than those of the informative packets at node $i$ under policy $\pi$, then the arrival times of the informative packets at node $j$ under policy $IP$ are earlier than those of the informative packets at node $j$ under policy $\pi$. Observe that the vector $\mathbf{\Gamma}_{i}(IP)$ represents the arrival times of the informative packets at node $i$ under policy $IP$ (recall the construction of the infeasible policy $IP$ and its age evolution in  \eqref{lower_bound_def_1}), while the vector $\mathbf{D}_{i}(\pi)$ represents the arrival times of the informative packets at node $i$ under policy $\pi$. Then, the previous statement is manifested in the following lemma. 
\begin{lemma}\label{lem_na2}
For any link $(i,j)\in\mathcal{L}$, if (i) the packet transmission times are NBU, and (ii) $\mathbf{\Gamma}_{i}(IP)\leq \mathbf{D}_{i}(\pi)$, then
\begin{equation}\label{NS1}
[\mathbf{\Gamma}_j(IP)\vert\mathcal{I}'] \leq_{\text{st}} [\mathbf{D}_j(\pi)\vert\mathcal{I}'],
\end{equation}
holds for all $\mathcal{I}'$ satisfying $B_{ij}\geq 1$.
\end{lemma}
\begin{proof}
See Appendix \ref{Appendix_E}
\end{proof}

\emph{Step 2}: We use Lemma \ref{lem_na2} to prove \eqref{Lower_bound_eq1}. Consider a node $j\in\mathcal{H}_k$. We prove \eqref{Lower_bound_eq1} using Theorem 6.B.3  and Theorem 6.B.16.(c) of \cite{shaked2007stochastic} into two steps:

\emph{Step A}: Consider node $i_{j,1}$. Observe that node $i_{j,1}$ receives update packets from node 0. Since the packet arrival times $(a_{10}, \ldots, a_{n0})$ to node 0 are invariant of the scheduling policy, both conditions of  Lemma \ref{lem_na2} are satisfied and we can apply it on the link $(0,i_{j,1})$ to obtain
\begin{equation}\label{eq_pr_1}
[\mathbf{\Gamma}_{i_{j,1}}(IP)\vert\mathcal{I}']\leq_{\text{st}} [\mathbf{D}_{i_{j,1}}(\pi)\vert\mathcal{I}'].
\end{equation}

\emph{Step B}: Consider node $i_{j,m}$, where $2\leq m\leq k$. We need to prove that 
\begin{equation}\label{eq_pr_2}
\begin{split}
&[\mathbf{\Gamma}_{i_{j,m}}(IP)\vert\mathcal{I}'\!, \mathbf{\Gamma}_{i_{j,1}}(IP)\!\!=\!\!\mathbf{\gamma}_{i_{j,1}},\ldots,\mathbf{\Gamma}_{i_{j,m-1}}(IP)\!\!=\!\!\mathbf{\gamma}_{i_{j,m-1}} ]\\&\leq_{\text{st}}\!\! [\mathbf{D}_{i_{j,m}}(\pi)\vert\mathcal{I}'\!,\mathbf{D}_{i_{j,1}}(\pi)\!\!=\!\!\mathbf{d}_{i_{j,1}},\ldots,\mathbf{D}_{i_{j,m-1}}(\pi)\!\!=\!\!\mathbf{d}_{i_{j,m-1}}],\\&
\text{whenever}\quad \mathbf{\gamma}_{i_{j,t}}\leq \mathbf{d}_{i_{j,t}},t=1, \ldots,m-1.
\end{split}
\end{equation}
Since node $i_{j,m}$ receives update packets from node $i_{j,m-1}$ and $\mathbf{\Gamma}_{i_{j,m-1}}(IP)\leq \mathbf{D}_{i_{j,m-1}}(\pi)$ in \eqref{eq_pr_2}, both conditions of Lemma \ref{lem_na2} are satisfied in this case as well (in particular, for the link $(i_{j,m-1},i_{j,m})$), and we can use it to prove \eqref{eq_pr_2}. By using \eqref{eq_pr_1} and \eqref{eq_pr_2} with Theorem 6.B.3 of \cite{shaked2007stochastic}, we can show
\begin{equation}\label{lower_bound_i_j_k}
[\mathbf{\Gamma}_{i_{j,1}}(IP), \ldots, \mathbf{\Gamma}_{i_{j,k}}(IP)\vert\mathcal{I}']\leq_{\text{st}} [\mathbf{D}_{i_{j,1}}(\pi), \ldots, \mathbf{D}_{i_{j,k}}(\pi)\vert\mathcal{I}'].
\end{equation}
 Following the previous argument, we can show that \eqref{lower_bound_i_j_k} holds for all $j\in\mathcal{V}$. Note that the transmission times are independent
across links. Using this with Theorem 6.B.3 and Theorem
6.B.16.(c) of \cite{shaked2007stochastic}, we prove \eqref{Lower_bound_eq1}. This completes the proof.

\section{Proof of Lemma \ref{lem_na2}}\label{Appendix_E}
The proof of Lemma \ref{lem_na2} is motivated by the proof idea of  \cite[Lemma 1]{sun2016delay} and \cite[Lemma 2]{sun2017near}. 
 \begin{figure*}[!tbp]
 \centering
 \subfigure[Case 1: Link $(i,j)$ sends packet $k$ after the time $\theta(IP)$.]{
  \includegraphics[scale=0.28]{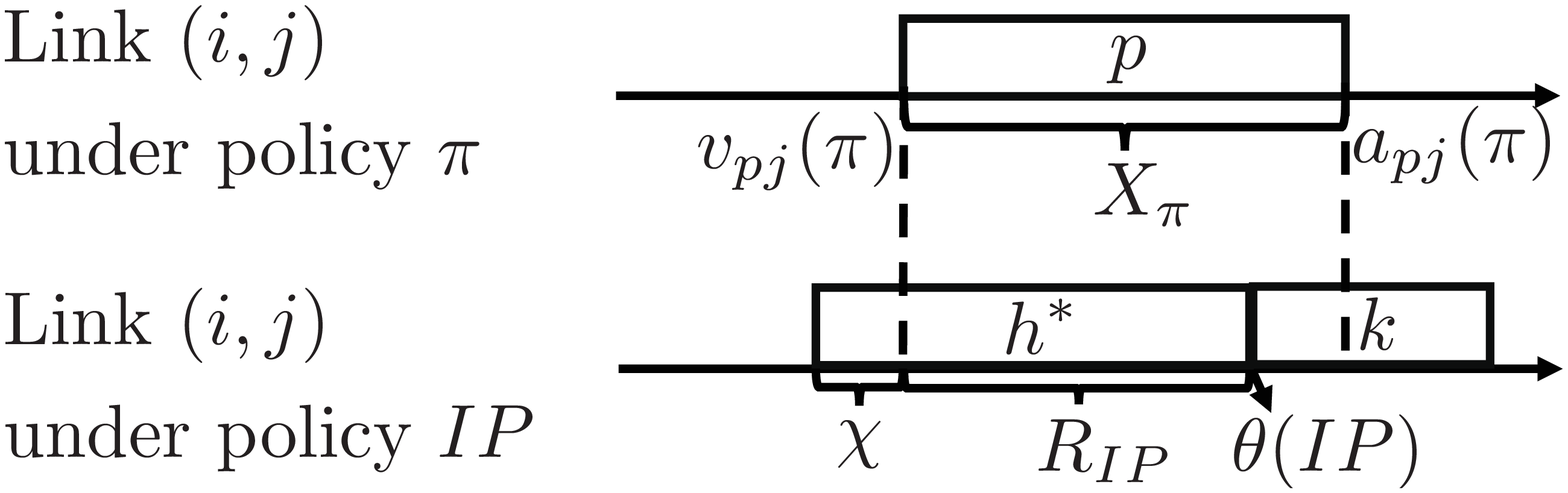}
   \label{a}
   }
 \subfigure[Case 2: Link $(i,j)$ sends packet $k$ before the time $IP$.]{
  \includegraphics[scale=0.28]{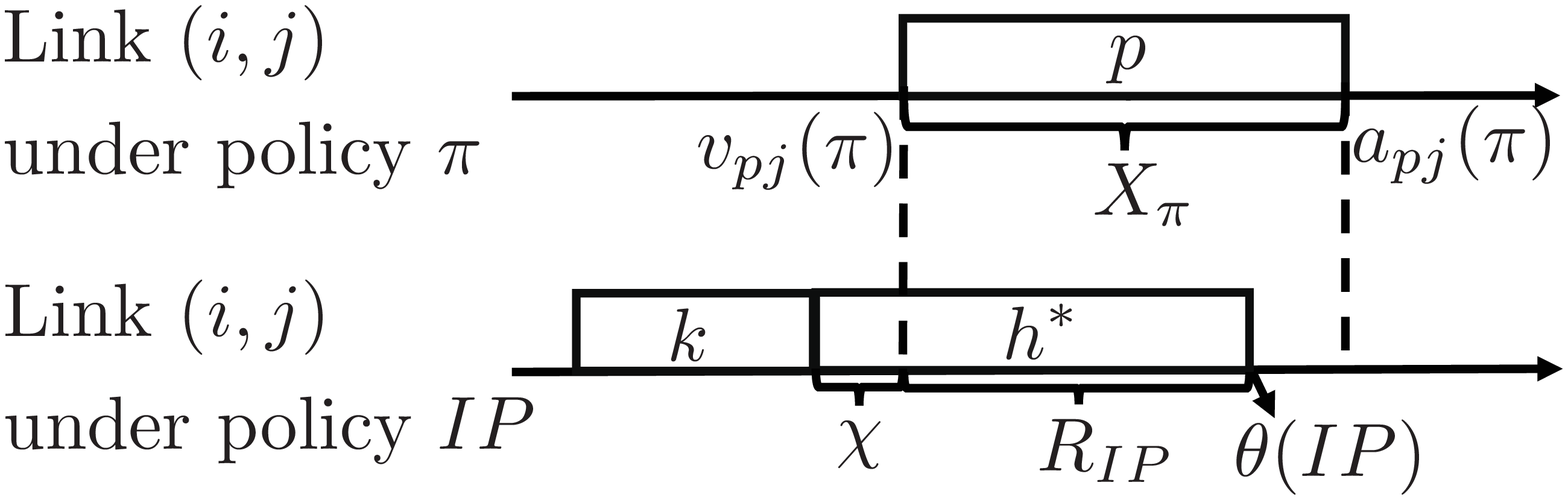}
   \label{b}
   }
\captionsetup{justification=justified, font={onehalfspacing}}
\caption{Illustration of packet transmissions under policy $\pi$ and policy $IP$. In policy $\pi$, link $(i,j)$ starts to send packet $p$ at time $v_{pj}(\pi)$ and will complete its transmission at time $a_{pj}(\pi)$. Hence, the transmission duration of packet $p$ is $[v_{pj}(\pi),a_{pj}(\pi)]$ in policy $\pi$. Under policy $IP$, link $(i,j)$ is kept busy before time $v_{pj}(\pi)$ for a time duration equal to the actual transmission time of packet $h^*$ and becomes available to send a new packet at the time $\theta(IP)<a_{pj}(\pi)$.} \label{fig:samp_path_NBU}
\end{figure*}
We prove \eqref{NS1} using Theorem 6.B.3 of  \cite{shaked2007stochastic} into two steps.

\emph{Step 1}: Consider packet 1. Note that packet 1 may not be the first packet to arrive at node $i$ under policy $IP$. Thus, we use $l^*$ to denote the index of the first arrived packet at node $i$ under policy $IP$, where $s_{l^*}\geq s_1$. From the construction of the policy $IP$ and \eqref{Def_na1_1}, $\Gamma_{1i}(IP)$ is the arrival time of the first arrived packet at node $i$ under policy $IP$. Since the link $(i,j)$ is idle before the arrival of the first arrived packet at node $i$, and policy $IP$ is a work-conserving policy, packet $l^*$ will start its transmission under policy $IP$ over the link $(i,j)$ once it arrives to node $i$ (at time $\Gamma_{1i}(IP)$). Thus, from \eqref{Def_na1_1}, we obtain
\begin{equation}\label{eq_11_1}
[\Gamma_{1j}(IP)\vert\mathcal{I}']=[v_{l^*j}(IP)\vert\mathcal{I}']=[\Gamma_{1i}(IP)\vert\mathcal{I}'].
\end{equation}
Observe that we have 
\begin{equation}\label{eq_11_2}
[\Gamma_{1i}(IP)\vert\mathcal{I}']\leq[D_{1i}(\pi)\vert\mathcal{I}'].
\end{equation}
 Also, we must have 
 \begin{equation}\label{eq_11_3}
  [D_{1i}(\pi)\vert\mathcal{I}']\leq[D_{1j}(\pi)\vert\mathcal{I}'],
 \end{equation}
 because a packet must spend a time over the link $(i,j)$ (its transmission time over the link $(i,j)$) before it is delivered from node $i$ to node $j$ under policy $\pi$. Combining \eqref{eq_11_1}, \eqref{eq_11_2}, and \eqref{eq_11_3}, we get 	 
\begin{equation}\label{prove1}
[\Gamma_{1j}(IP)\vert\mathcal{I}']=[\Gamma_{1i}(IP)\vert\mathcal{I}']\leq[D_{1i}(\pi)\vert\mathcal{I}']\leq[D_{1j}(\pi)\vert\mathcal{I}'].
\end{equation}

\emph{Step 2}: Consider a packet $p$, where $2\leq p \leq n$.  We suppose that no packet with generation time greater than $s_{p}$ has arrived to node $j$ before packet $p$ under policy $\pi$. We need to prove that
\begin{equation}\label{NBU_to_proof}
\begin{split}
[\Gamma_{pj}(IP)\vert\mathcal{I}',\Gamma_{1j}(IP)=\gamma_1, \ldots, \Gamma_{(p-1)j}(IP)=\gamma_{p-1}]\\
\leq_{\text{st}}[D_{pj}(\pi)\vert\mathcal{I}', D_{1j}(\pi)=d_1, \ldots, D_{(p-1)j}(\pi)=d_{p-1}],\\
\text{whenever}\quad \gamma_l\leq d_l, l=1, 2, \ldots, p-1. 
\end{split}
\end{equation}
For notation simplicity, define $\Gamma^{p-1}\triangleq \{\Gamma_{1j}(IP)=\gamma_1, \ldots, \Gamma_{(p-1)j}(IP)=\gamma_{p-1}\}$ and $D^{p-1}\triangleq \{ D_{1j}(\pi)=d_1, \ldots, D_{(p-1)j}(\pi)=d_{p-1}\}$. We will show that the link $(i,j)$ under policy $IP$ can send a new packet at a time that is stochastically smaller than the arrival time of packet $p$ at node $j$ under policy $\pi$.  At this time, there are two possible cases under policy $IP$. One of them is that the link $(i,j)$ sends a packet with generation time greater than $s_p$. The other one is that the link $(i,j)$ sends a packet with generation time less than $s_p$ or there is no packet to be sent. We will show that \eqref{NBU_to_proof} holds in either case.

As illustrated in Fig. \ref{fig:samp_path_NBU}, suppose that under policy $\pi$, link $(i,j)$ starts to send packet $p$ at time $v_{pj}(\pi)$ and will complete its transmission at time $a_{pj}(\pi)$. Under policy $IP$, define $h^*=\argmax_{h}\{v_{hj}(IP) : v_{hj}(IP)\leq v_{pj}(\pi)\}$ as the index of the last packet whose transmission starts over the link $(i,j)$ before time $v_{pj}(\pi)$. Note that the link $(i,j)$ under policy $IP$ is kept busy after time $v_{h^*j}(IP)$ for a time duration equal to the actual transmission time of packet $h^*$ over the link $(i,j)$.
 Suppose that under policy $IP$, link $(i,j)$ is kept busy for $\chi$ $(\chi\geq 0)$ seconds of the actual transmission time of packet $h^*$ before time $v_{pj}(\pi)$. Let $R_{IP}$ denote the remaining busy period of the link $(i,j)$ under policy $IP$ after time $v_{pj}(\pi)$ (this remaining busy period is due to the remaining transmission time of packet $h^*$ after time $v_{pj}(\pi)$). Hence, link $(i,j)$ becomes available to send a new packet at time $v_{pj}(\pi)+R_{IP}$. Let $X_{\pi}=a_{pj}(\pi)-v_{pj}(\pi)$ denote the transmission time of packet $p$ under policy $\pi$ and $X_{IP}=\chi+R_{\text{LB}}$ denote the actual transmission time of packet $h^*$. Then, the CCDF of $R_{IP}$ is given by
 \begin{equation}\label{stoch1}
\begin{split}
\mathbb{P}[R_{IP}> s]=\mathbb{P}[X_{IP}-\chi>s\vert X_{IP}>\chi].
\end{split}
\end{equation}
Because the packet transmission times are NBU, we can obtain that for all $s,\chi\geq 0$
 \begin{equation}\label{stoch2}
\begin{split}
\mathbb{P}[X_{IP}-\chi>s\vert X_{IP}>\chi]&=
\mathbb{P}[X_{\pi}-\chi>s\vert X_{\pi}>\chi]\\&\leq \mathbb{P}[X_{\pi}>s].
\end{split}
\end{equation}
By combining \eqref{stoch1} and \eqref{stoch2}, we obtain
 \begin{equation}
\begin{split}
R_{IP}\leq_{\text{st}}X_{\pi},
\end{split}
\end{equation}
which implies
\begin{equation}\label{R-X}
v_{pj}(\pi)+R_{IP}\leq_{\text{st}}v_{pj}(\pi)+X_{\pi}=a_{pj}(\pi).
\end{equation}
From \eqref{R-X}, we can deduce that link $(i,j)$ becomes available to send a new packet under policy $IP$ at a time that is stochastically smaller than the time $a_{pj}(\pi)$. Let $\theta(IP)$ denote the time that link $(i,j)$ becomes available to send a new packet under policy $IP$. According to \eqref{R-X}, we have
 \begin{equation}
\begin{split}
[\theta(IP)\vert\mathcal{I}', \Gamma^{p-1}]\leq_{\text{st}}[a_{pj}(\pi)\vert\mathcal{I}', D^{p-1}],\\ \text{whenever}\quad \gamma_l\leq d_l, l=1, 2, \ldots, p-1.
\end{split}
\end{equation}
It is important to note that, since we have  $[\Gamma_{pi}(IP)\vert\mathcal{I}']\leq[D_{pi}(\pi)\vert\mathcal{I}']$, there is a packet with generation time greater than $s_p$ is available to the link $(i,j)$ before time $v_{pj}(\pi)$ under policy $IP$. At the time $\theta(IP)$, we have two possible cases under policy $IP$:

Case 1: Link $(i,j)$ starts to send a fresh packet $k$ with $k\geq p$ at the time $\theta(IP)$ under policy $IP$, as shown in Fig. \ref{a}. Hence we obtain 
 \begin{equation}\label{case1_1}
\begin{split}
[v_{kj}(IP)\vert\mathcal{I}', \Gamma^{p-1}]=[\theta(IP)\vert\mathcal{I}', \Gamma^{p-1}]\leq_{\text{st}}[a_{pj}(\pi)\vert\mathcal{I}', D^{p-1}]\\ \text{whenever}\quad \gamma_l\leq d_l, l=1, 2, \ldots, p-1.
\end{split}
\end{equation}
Since $s_k\geq s_p$, \eqref{Def_na1_1} implies
\begin{equation}\label{case1_2}
[\Gamma_{pj}(IP)\vert\mathcal{I}', \Gamma^{p-1}]\leq [v_{kj}(IP)\vert\mathcal{I}', \Gamma^{p-1}].
\end{equation}
Since there is no packet with generation time greater than $s_p$ that has been arrived to node $j$ before packet $p$ under policy $\pi$, \eqref{Def_na1_2} implies 
\begin{equation}\label{case1_3}
[D_{pj}(\pi)\vert\mathcal{I}', D^{p-1}]=[a_{pj}(\pi)\vert\mathcal{I}', D^{p-1}].
\end{equation}
By combining \eqref{case1_1}, \eqref{case1_2}, and \eqref{case1_3}, \eqref{NBU_to_proof} follows.

Case 2: Link $(i,j)$ starts to send a stale packet (with generation time smaller than $s_p$) or there is no packet transmission over the link $(i,j)$ at the time $\theta(\text{LB})$ under policy $IP$. Since the packets are served by following a work-conserving LGFS principle under policy $IP$, and a packet with generation time greater than $s_p$ is available to the link $(i,j)$ before time $v_{pj}(\pi)$ under policy $IP$, the link $(i,j)$ must have sent a fresh packet $k$ with $k\geq p$ before time $\theta(IP)$, as shown in Fig. \ref{b}. Hence, we have
 \begin{equation}\label{case2_1}
\begin{split}
[v_{kj}(IP)\vert\mathcal{I}', \Gamma^{p-1}]\leq[\theta(IP)\vert\mathcal{I}', \Gamma^{p-1}]\leq_{\text{st}}[a_{pj}(\pi)\vert\mathcal{I}', D^{p-1}]\\ \text{whenever}\quad \gamma_l\leq d_l, l=1, 2, \ldots, p-1.
\end{split}
\end{equation}
Similar to Case 1, we can use \eqref{Def_na1_1}, \eqref{Def_na1_2}, and \eqref{case2_1} to show that \eqref{NBU_to_proof} holds in this case. 

Notice that if there is a fresher packet $y$ with $s_y>s_{p}$ and $a_{yj}(\pi)<a_{pj}(\pi)$ (this may occur if packet $y$ preempts the transmission of packet $p$ under policy $\pi$ or packet $y$ arrives to node $i$ before packet $p$ under policy $\pi$), then we replace packet $p$ by packet $y$ in the arguments and equations from \eqref{NBU_to_proof} to \eqref{case2_1} to obtain
 \begin{equation}\label{casey_1}
\begin{split}
[\Gamma_{yj}(IP)\vert\mathcal{I}', \Gamma^{p-1}]\leq[D_{yj}(\pi)\vert\mathcal{I}', D^{p-1}]\\ \text{whenever}\quad \gamma_l\leq d_l, l=1, 2, \ldots, p-1.
\end{split}
\end{equation}
Observing that $s_y>s_p$, \eqref{Def_na1_1} implies
\begin{equation}\label{casey_2}
[\Gamma_{pj}(IP)\vert\mathcal{I}', \Gamma^{p-1}]\leq[\Gamma_{yj}(IP)\vert\mathcal{I}', \Gamma^{p-1}].
\end{equation}
Since $a_{yj}(\pi)<a_{pj}(\pi)$ and $s_y>s_{p}$, \eqref{Def_na1_2} implies
\begin{equation}\label{casey_3}
[D_{pj}(\pi)\vert\mathcal{I}', D^{p-1}]=[D_{yj}(\pi)\vert\mathcal{I}', D^{p-1}].
\end{equation}
By combining \eqref{casey_1}, \eqref{casey_2}, and \eqref{casey_3}, we can prove \eqref{NBU_to_proof} in this case too. Finally, substitute \eqref{prove1} and \eqref{NBU_to_proof} into Theorem 6.B.3 of  \cite{shaked2007stochastic}, \eqref{NS1} is proven. 
%

\section{Proof of Theorem \ref{thmnbu_gab}}\label{Appendix_D}
\begin{figure}
\centering
\includegraphics[scale=0.35]{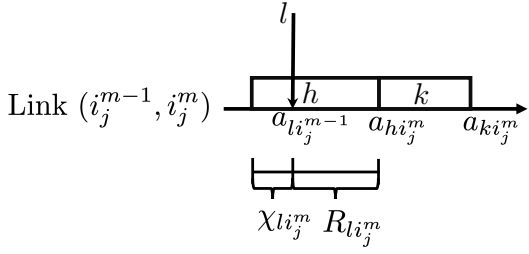}
\centering
\captionsetup{justification=justified, font={onehalfspacing}}
\caption{ An illustration of $R_{li_{j,m}}$ and $\chi_{li_{j,m}}$. Packet $l$ arrives to node $i_{j,m-1}$ at time $a_{li_{j,m-1}}$, while packet $h$ with $h<l$ is being transmitted over the link $(i_{j,m-1},i_{j,m})$. After the delivery of packet $h$ to node $i_{j,m}$ at time $a_{hi_{j,m}}$, packet $k$ with $k\geq l$ is transmitted over the link $(i_{j,m-1},i_{j,m})$. The duration $R_{li_{j,m}}$ is the waiting time of packet $l$ in the queue of the link $(i_{j,m-1},i_{j,m})$ until the packet $k$ starts its transmission. The duration $\chi_{li_{j,m}}$ is the time spent by the link $(i_{j,m-1},i_{j,m})$ on sending the packet $h$ before the time $a_{li_{j,m-1}}$.}\label{Fig:parameter2}
\vspace{-0.3cm}
\end{figure}

For notation simplicity, let policy $P$ represent the non-prmp-LGFS policy and policy $IP$ represent the infeasible policy (the construction of the infeasible policy and the age lower bound are provided in Appendix \ref{Appendix_A'}). We will need the definitions that are provided at the beginning of Appendix \ref{Appendix_C} throughout this proof. Consider a node $j\in\mathcal{H}_k$ with $k\geq 1$. We prove Theorem \ref{thmnbu_gab} into three steps:

\emph{Step 1:} We provide an upper bound on the time differences between the arrival times of the informative packets at node $j$ under policy $IP$ and those under policy $P$. To achieve that, we need the following definitions. For each link in the path to node $j$ (i.e., $(i_{j,m-1},i_{j,m})$ for all $1\leq m\leq k$), define $R_{li_{j,m}}=\Gamma_{li_{j,m}}-D_{li_{j,m-1}}$ as the time spent in the queue of the link $(i_{j,m-1},i_{j,m})$ by the packet that arrives at node $i_{j,m-1}$ at time $D_{li_{j,m-1}}$, until the first transmission starting time over the link $(i_{j,m-1},i_{j,m})$ of the packets with generation time greater than $s_l$. If there is a packet that is being transmitted over the link $(i_{j,m-1},i_{j,m})$ at time $D_{li_{j,m-1}}$, let $\chi_{li_{j,m}}$ $(\chi_{li_{j,m}}\geq 0)$ denote the amount of time that the link $(i_{j,m-1},i_{j,m})$ has spent on sending this packet by the time $D_{li_{j,m-1}}$. These parameters ($R_{li_{j,m}}$ and $\chi_{li_{j,m}}$) are functions of the scheduling policy $\pi$. An illustration of these parameters is provided in Fig. \ref{Fig:parameter2}. Note that policy $P$ is a LGFS work-conserving policy. Also, the packets under policy $IP$ are served by following a work-conserving LGFS principle. Thus, we can express $R_{li_{j,m}}$ under these policies as $R_{li_{j,m}}=[X_{i_{j,m}}-\chi_{li_{j,m}}\vert X_{i_{j,m}}>\chi_{li_{j,m}}]$. Because the packet transmission times are NBU and \emph{i.i.d.} across time, for all realization of $\chi_{li_{j,m}}$
\begin{equation}
[R_{li_{j,m}}\vert \chi_{li_{j,m}}]\leq_{\text{st}}X_{i_{j,m}},~ \text{for}~ m=1,\ldots,k,~\forall l,
\end{equation}
which implies that 
\begin{equation}\label{rem_les_eq}
\mathbb{E}[R_{li_{j,m}}\vert \chi_{li_{j,m}}]\leq\mathbb{E}[X_{i_{j,m}}], ~\text{for}~m=1,\ldots,k,~\forall l,
\end{equation}
holds for policy $P$ and policy $IP$. Define $z_l=D_{lj}(P)-\Gamma_{lj}(IP)$. Note that, $\Gamma_{lj}(IP)$ represents the arrival time at node $j$ of a packet $p$ with $s_p\geq s_l$ under policy $IP$, and $D_{lj}(P)$ represents the arrival time at node $j$ of a packet $h$ with $s_h\geq s_l$ under policy $P$. Therefor, $z_l$'s represent the time differences between the arrival times of the informative packets at node $j$  under policy $IP$ and those under policy $P$, as shown in Fig. \ref{Fig:Gap2}. By invoking the construction of policy $IP$, we have $D_{lj}(IP)=\Gamma_{lj}(IP)$ for all $l$. Using this with the definition of $R_{li_{j,m}}$, we can express $\Gamma_{lj}(IP)$ as
\begin{equation}\label{eq_st2_1}
\Gamma_{lj}(IP)=a_{l0}+\sum_{m=1}^k[R_{li_{j,m}}(IP)\vert\chi_{li_{j,m}}(IP)],
\end{equation}
 where $\Gamma_{lj}(IP)$ is considered as the arrival time at node $j$ of the first packet with generation time greater than $s_l$  under policy $IP$. Also, we can express $D_{lj}(P)$ as 
\begin{equation}\label{eq_st2_2}
D_{lj}(P)=a_{l0}+\sum_{m=1}^k[R_{li_{j,m}}(P)\vert\chi_{li_{j,m}}(P)]+\sum_{m=1}^kX_{i_{j,m}}.
\end{equation}
Observing that packet arrival times $(a_{10}, a_{20}, \ldots)$ at node $0$ and the packet transmission times are invariant of the scheduling policy $\pi$. Then, from the construction of policy $IP$, we have $[R_{li_{j,1}}(IP)\vert\chi_{li_{j,1}}(IP)]=[R_{li_{j,1}}(P)\vert\chi_{li_{j,1}}(P)]$ for all $l$ (because all nodes in $\mathcal{H}_1$ receive the update packets from node 0). Using this with \eqref{eq_st2_1} and \eqref{eq_st2_2}, we can obtain
\begin{equation}\label{z-eqn}
\begin{split}
\!\!\!\!\!\!\!\!z_l&=D_{lj}(P)-\Gamma_{lj}(IP)\\&=\sum_{m=2}^k[R_{li_{j,m}}(P)\vert\chi_{li_{j,m}}(P)]+\sum_{m=1}^kX_{i_{j,m}}\\&-\sum_{m=2}^k[R_{li_{j,m}}(IP)\vert\chi_{li_{j,m}}(IP)]\\&\leq \sum_{m=2}^k[R_{li_{j,m}}(P)\vert\chi_{li_{j,m}}(P)]+\sum_{m=1}^kX_{i_{j,m}}=z_l^{'}.
\end{split}
\end{equation}
Since the packet transmission times are independent of the packet generation process, we also have $z_l^{'}$'s are independent of the packet generation process. In addition, from \eqref{rem_les_eq}, we have
\begin{equation}
\mathbb{E}[z_l^{'}]\leq \mathbb{E}[X_{i_{j,1}}]+2\sum_{m=2}^k\mathbb{E}[X_{i_{j,m}}].
\end{equation}
\begin{figure}
\includegraphics[scale=0.4]{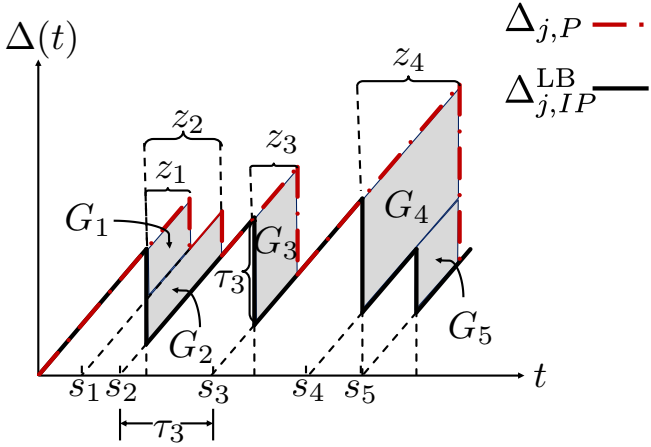}
\centering
\caption{The evolution $\Delta_{j,IP}^{\text{LB}}$ and $\Delta_{j,P}$.}\label{Fig:Gap2}
\vspace{-0.3cm}
\end{figure}

\emph{Step 2:} We use Step 1 to provide an upper bound on the average gap between $\Delta_{j,IP}^{\text{LB}}$ and $\Delta_{j,P}$. This gap process is denoted by $\{G_j(t), t\in[0,\infty)\}$. The average gap is given by 
\begin{equation}\label{gapeq1}
[\bar{G}_j\vert\mathcal{I}']=\limsup_{T\rightarrow\infty}\frac{\int_0^{T} G_j(t)dt}{T}.
\end{equation}
Let $\tau_l$ denote the inter-generation time between packet $l$ and packet $l-1$ (i.e., $\tau_l=s_{l}-s_{l-1}$), where $\tau=\{\tau_l, l\geq 1\}$. Define $N(T)=\max\{l : s_{l}\leq T\}$ as the number of generated packets by time $T$. Note that $[0,s_{N(T)}]\subseteq [0,T]$, where the length of the interval 	$[0,s_{N(T)}]$ is $\sum_{l=1}^{N(T)}\tau_l$. Thus, we have 
\begin{equation}\label{Igab_eq_1}
\sum_{l=1}^{N(T)}\tau_l\leq T.
\end{equation}
The area defined by the integral in \eqref{gapeq1} can be decomposed into a sum of disjoint geometric parts. Observing Fig. \ref{Fig:Gap2}, the area can be approximated by the concatenation of the parallelograms $G_{1}, G_{2},\ldots$ ($G_l$'s are highlighted in Fig. \ref{Fig:Gap2}). Note that the parallelogram $G_{l}$ results after the generation of packet $l$ (i.e., the gap that is corresponding to the packet $l$, occurs after its generation). Since the observing time $T$ is chosen arbitrary, when $T\geq s_{l}$, the total area of the parallelogram $G_{l}$ is accounted in the summation $\sum_{l=1}^{N(T)}G_{l}$, while it may not be accounted in the integral $\int_{0}^T G_f(t)dt$. This implies that
\begin{equation}\label{Igab_eq_2}
\sum_{l=1}^{N(T)}G_{l}\geq \int_{0}^T G_j(t)dt.
\end{equation}
Combining \eqref{Igab_eq_1} and \eqref{Igab_eq_2}, we get 
\begin{equation}\label{Igab_eq_3}
\frac{\int_0^{T} G_j(t)dt}{T}\leq\frac{\sum_{l=1}^{N(T)}G_{l}}{\sum_{l=1}^{N(T)}\tau_l}.
\end{equation}
Then, take conditional expectation given $\tau$ and $N(T)$ on both sides of \eqref{Igab_eq_3}, we obtain
\begin{equation}\label{Igab_eq_4}
\begin{split}
\frac{\mathbb{E}[\int_0^{T} G_j(t)dt\vert\tau, N(T)]}{T}&\leq\frac{\mathbb{E}[\sum_{l=1}^{N(T)}G_{l}\vert\tau, N(T)]}{\sum_{l=1}^{N(T)}\tau_l}\\&=\frac{\sum_{l=1}^{N(T)}\mathbb{E}[G_{l}\vert\tau, N(T)]}{\sum_{l=1}^{N(T)}\tau_l},
\end{split}
\end{equation}
where the second equality follows from the linearity of the expectation. From Fig. \ref{Fig:Gap2}, $G_{l}$ can be calculated as 
\begin{equation}\label{Igab_eq_5}
G_{l}=\tau_lz_{l}.
\end{equation}
substituting by \eqref{Igab_eq_5} into \eqref{Igab_eq_4}, yields
\begin{equation}
\begin{split}
\frac{\mathbb{E}[\int_0^{T} G_j(t)dt\vert\tau, N(T)]}{T}&\leq\frac{\sum_{l=1}^{N(T)}\mathbb{E}[\tau_lz_{l}\vert\tau, N(T)]}{\sum_{l=1}^{N(T)}\tau_l}\\&=\frac{\sum_{l=1}^{N(T)}\tau_l\mathbb{E}[z_{l}\vert\tau, N(T)]}{\sum_{l=1}^{N(T)}\tau_l}.
\end{split}
\end{equation}
Using \eqref{z-eqn}, we obtain
 \begin{equation}\label{Igab_eq_6}
 \begin{split}
\frac{\mathbb{E}[\int_0^{T} G_j(t)dt\vert\tau, N(T)]}{T}&\leq\frac{\sum_{l=1}^{N(T)}\tau_l\mathbb{E}[z_{l}\vert\tau, N(T)]}{\sum_{l=1}^{N(T)}\tau_l}\\&\leq\frac{\sum_{l=1}^{N(T)}\tau_l\mathbb{E}[z_{l}^{'}\vert\tau, N(T)]}{\sum_{l=1}^{N(T)}\tau_l}.
\end{split}
\end{equation}
Note that $z_{l}^{'}$'s are independent of the packet generation process. Thus, we have $\mathbb{E}[z_{l}^{'}\vert\tau, N(T)]=\mathbb{E}[z_{l}^{'}]\leq\mathbb{E}[X_{i_{j,1}}]+2\sum_{m=2}^k\mathbb{E}[X_{i_{j,m}}]$ for all $l$. Using this in \eqref{Igab_eq_6}, we get
\begin{equation*}
\begin{split}
\frac{\mathbb{E}[\int_0^{T} G_j(t)dt\vert\tau, N(T)]}{T}&\leq\frac{\sum_{l=1}^{N(T)}\tau_l(\mathbb{E}[X_{i_{j,1}}]+2\sum_{m=2}^k\mathbb{E}[X_{i_{j,m}}])}{\sum_{l=1}^{N(T)}\tau_l}\\&\leq \mathbb{E}[X_{i_{j,1}}]+2\sum_{m=2}^k\mathbb{E}[X_{i_{j,m}}],
\end{split}
\end{equation*}
by the law of iterated expectations, we have
\begin{equation}\label{Igab_eq_7}
\frac{\mathbb{E}[\int_0^{T} G_j(t)dt]}{T}\leq \mathbb{E}[X_{i_{j,1}}]+2\sum_{m=2}^k\mathbb{E}[X_{i_{j,m}}].
\end{equation}
Taking $\limsup$ of both sides of \eqref{Igab_eq_7} when $T\rightarrow\infty$, yields
\begin{equation}\label{Igab_eq_8}
\limsup_{T\rightarrow\infty}\frac{\mathbb{E}[\int_0^{T} G_j(t)dt]}{T}\leq \mathbb{E}[X_{i_{j,1}}]+2\sum_{m=2}^k\mathbb{E}[X_{i_{j,m}}].
\end{equation}
Equation \eqref{Igab_eq_8} tells us that the average gap between $\Delta_{f,IP}^{\text{LB}}$ and $\Delta_{f,P}$ is no larger than $\mathbb{E}[X_{i_{j,1}}]+2\sum_{m=2}^k\mathbb{E}[X_{i_{j,m}}]$.

\emph{Step 3:} We use the provided upper bound on the gap in Step
2 to prove \eqref{gap_main1}. Since $\Delta_{j,IP}^{\text{LB}}$ is a lower bound of $\Delta_{j,P}$, we obtain
\begin{equation}\label{Igab_eq_9}
\begin{split}
&[\bar{\Delta}_{j,IP}^{\text{LB}}\vert\mathcal{I}']\leq [\bar{\Delta}_{j,P}\vert\mathcal{I}']\leq\\& [\bar{\Delta}_{j,IP}^{\text{LB}}\vert\mathcal{I}']+\mathbb{E}[X_{i_{j,1}}]+2\sum_{m=2}^k\mathbb{E}[X_{i_{j,m}}],
\end{split}
\end{equation}
where $\bar{\Delta}_{j,IP}^{\text{LB}}=\limsup_{T\rightarrow\infty}\frac{\mathbb{E}[\int_0^{T} \Delta_{j,IP}^{\text{LB}}(t)dt]}{T}$. From Lemma \ref{lem_lower_bound_all} in Appendix \ref{Appendix_A'}, we have for all $\mathcal{I}'$ satisfying $B_{ij}\geq 1$, and $\pi\in\Pi$
\begin{equation}
[\Delta_{j,IP}^{\text{LB}}\vert\mathcal{I}']\leq_{\text{st}}[\Delta_{j,\pi}\vert\mathcal{I}'],
\end{equation}
which implies that
\begin{equation}
[\bar{\Delta}_{j,IP}^{\text{LB}}\vert\mathcal{I}']\leq[\bar{\Delta}_{j,\pi}\vert\mathcal{I}'],
\end{equation}
holds for all $\pi\in\Pi$. As a result, we get
\begin{equation}\label{Igab_eq_10}
[\bar{\Delta}_{j,IP}^{\text{LB}}\vert\mathcal{I}']\leq\min_{\pi\in\Pi}[\bar{\Delta}_{j,\pi}\vert\mathcal{I}'].
\end{equation}
Since policy $P$ is a feasible policy, we get
\begin{equation}\label{Igab_eq_11}
\min_{\pi\in\Pi}[\bar{\Delta}_{j,\pi}\vert\mathcal{I}']\leq [\bar{\Delta}_{j,P}\vert\mathcal{I}'].
\end{equation}
Combining \eqref{Igab_eq_9}, \eqref{Igab_eq_10}, and \eqref{Igab_eq_11}, we get
\begin{equation}\label{Igab_eq_12'}
\begin{split}
&\min_{\pi\in\Pi}[\bar{\Delta}_{j,\pi}\vert\mathcal{I}']\leq [\bar{\Delta}_{j,P}\vert\mathcal{I}']\leq\\& \min_{\pi\in\Pi}[\bar{\Delta}_{j,\pi}\vert\mathcal{I}']+\mathbb{E}[X_{i_{j,1}}]+2\sum_{m=2}^k\mathbb{E}[X_{i_{j,m}}].
\end{split}
\end{equation}
Following the previous argument, we can show that \eqref{Igab_eq_12'} holds for all $j\in\mathcal{V}\backslash \{0\}$. This proves \eqref{gap_main1}, which completes the proof.

\section{Proof of Theorem \ref{thm2}}\label{Appendix_B}
This proof is similar to that of Theorem \ref{thm1}. The difference between this proof and the proof of Theorem \ref{thm1} is that policy $\pi$ cannot be a preemptive policy here. We will use the same definition of the system state of policy $\pi$ used in Theorem \ref{thm1}. For notational simplicity, let policy $P$ represent the non-preemptive LGFS policy.

The key step in the proof of Theorem \ref{thm2} is the following lemma, where we compare policy $P$ with an arbitrary policy $\pi\in\Pi_{npwc}$.
 \begin{lemma}\label{lem2np}
 Suppose that $\mathbf{U}_{P}(0^-)=\mathbf{U}_{\pi}(0^-)$ for all $\pi\in\Pi_{npwc}$, then for all $\mathcal{I}$,
\begin{equation}\label{law9np}
\begin{split}
[\{\mathbf{U}_{P}(t),  t\in[0,\infty)\}\vert\mathcal{I}]\!\geq_{\text{st}}\! [\{\mathbf{U}_{\pi}(t), t\in[0,\infty)\}\vert\mathcal{I}].
 \end{split}
\end{equation}
\end{lemma}

 We use coupling and forward induction to prove Lemma \ref{lem2np}.
For any work-conserving policy $\pi$, suppose that stochastic processes $\widetilde{\mathbf{U}}_{P}(t)$ and $\widetilde{\mathbf{U}}_{\pi}(t)$ have the same distributions with $\mathbf{U}_{P}(t)$ and $\mathbf{U}_{\pi}(t)$, respectively. 
The state processes $\widetilde{\mathbf{U}}_{{P}}(t)$ and $\widetilde{\mathbf{U}}_{\pi}(t)$ are coupled in the following manner: If a packet is delivered from node $i$ to node $j$ at time $t$ as $\widetilde{\mathbf{U}}_{{P}}(t)$ evolves in policy prmp-LGFS,  then there exists a packet delivery from node $i$ to node $j$ at time $t$ as $\widetilde{\mathbf{U}}_{\pi}(t)$ evolves in policy $\pi$.
Such a coupling is valid since the transmission time distribution at each link is identical under all policies. Moreover, policy $\pi$ can not be either preemptive or non-work-conserving policy, and both policies have the same packets generation times $(s_1, s_2, \ldots, s_n)$ at the exterenal source and packet arrival times $(a_{10}, a_{20}, \ldots, a_{n0})$ to node 0. According to Theorem 6.B.30 in \cite{shaked2007stochastic}, if we can show 
\begin{equation}\label{main_eqnp}
\begin{split}
\mathbb{P}[\widetilde{\mathbf{U}}_{P}(t)\geq\widetilde{\mathbf{U}}_{\pi}(t), t\in[0,\infty)\vert\mathcal{I}]=1,
\end{split}
\end{equation}
then \eqref{law9np} is proven.  

To ease the notational burden, we will omit the tildes henceforth on the coupled versions and just use $\mathbf{U}_{P}(t)$ and $\mathbf{U}_{\pi}(t)$.

Next, we use the following lemmas to prove \eqref{main_eqnp}:

\begin{lemma}\label{lem3np}
Suppose that under policy $P$, $\mathbf{U}_{P}(\nu)$ is obtained by a packet delivery over the link $(i,j)$ at time $\nu$ in the system whose state is $\mathbf{U}_{P}(\nu^{-})$. Further, suppose that under policy $\pi$, $\mathbf{U}_{\pi}(\nu)$ is obtained by a packet delivery over the link $(i,j)$ at time $\nu$ in the system whose state is $\mathbf{U}_\pi(\nu^{-})$. If
\begin{equation}\label{hyp1np}
\begin{split}
\mathbf{U}_{P}(t)& \geq \mathbf{U}_\pi(t),
\end{split}
\end{equation}
holds for all $t\in [0, \nu^{-}]$, then
\begin{equation}\label{law6np}
\mathbf{U}_{P}(\nu) \geq \mathbf{U}_{\pi}(\nu).
\end{equation}
\end{lemma}

\begin{proof}
Let $s_{P}$ and $s_\pi$ denote the packet indexes and the generation times of the delivered packets over the link $(i,j)$ at time $\nu$ under policy $P$ and policy $\pi$, respectively.  From the definition of the system state, we can deduce that
\begin{equation}\label{Def1np}
\begin{split}
U_{j,P}(\nu)&=\max\{U_{j,P}(\nu^{-}),s_{P}\},\\
U_{j,\pi}(\nu)&=\max\{U_{j,\pi}(\nu^{-}),s_{\pi}\}.
\end{split}
\end{equation}
Hence, we have two cases:

Case 1: If $s_{P}\geq s_\pi$. From \eqref{hyp1np}, we have
\begin{equation}\label{pf21}
U_{j,P}(\nu^{-})\geq U_{j,\pi}(\nu^{-}).
\end{equation}
By $s_{P}\geq s_\pi$, \eqref{Def1np}, and \eqref{pf21}, we have
\begin{equation}
U_{j,P}(\nu)\geq U_{j,\pi}(\nu).
\end{equation}
Since there is no packet delivery under other links, we get
\begin{equation}
\begin{split}
U_{k,P}(\nu)=U_{k,P}(\nu^{-})\geq U_{k,\pi}(\nu^{-})=U_{k,\pi}(\nu), \quad\forall k\neq j.
\end{split}
\end{equation}
Hence, we have 
\begin{equation}
\mathbf{U}_{P}(\nu) \geq \mathbf{U}_{\pi}(\nu).
\end{equation}


Case 2: If $s_{P}<s_\pi$. Let $a_\pi$ represent the arrival time of packet $s_\pi$ to node $i$ under policy $\pi$. The transmission starting time of the delivered packets over the link $(i,j)$ is denoted by $\tau$ under both policies. Apparently, $a_\pi\leq\tau\leq\nu^{-}$. Since  packet $s_\pi$ arrived to node $i$  at time $a_\pi$ in policy $\pi$, we get
\begin{equation}\label{pf23}
s_\pi\leq U_{i,\pi}(a_\pi).
\end{equation}
From \eqref{hyp1np}, we obtain 
\begin{equation}\label{pf24}
U_{i,\pi}(a_\pi)\leq U_{i,P}(a_\pi).
\end{equation}
Combining \eqref{pf23} and \eqref{pf24}, yields
\begin{equation}\label{pf25}
s_\pi\leq U_{i,P}(a_\pi).
\end{equation}
Hence, in policy $P$, node $i$ has a packet with generation time no smaller than $s_\pi$ by the time $a_\pi$. Because the $U_{i,P}(t)$ is a non-decreasing function of $t$ and $a_\pi\leq \tau$, we have
\begin{equation}\label{pf26}
U_{i,P}(a_\pi)\leq U_{i,P}(\tau).
\end{equation}
Then, \eqref{pf25} and \eqref{pf26} imply
\begin{equation}\label{pf26'}
s_\pi\leq U_{i,P}(\tau).
\end{equation}
Since $s_{P} < s_\pi$, \eqref{pf26'} tells us
\begin{equation}\label{pf26''}
s_{P} < U_{i,P}(\tau),
\end{equation}
and hence policy $P$ is sending a stale packet on link $(i,j)$. By the definition of policy $P$, this happens only when all packets that are generated after $s_{P}$ in the queue of the link $(i,j)$ have been delivered to node $j$ by time $\tau$. In addition, \eqref{pf26'} tells us that by time $\tau$, node $i$ has already received a packet (say packet $h$) generated no earlier than $s_\pi$ in policy $P$. By $s_{P} < s_\pi$, packet $h$ is generated after $s_{P}$. Hence, packet $h$ must have been delivered to node $j$ by time $\tau$ in policy $P$ such that 
\begin{equation}\label{pf27}
s_\pi\leq U_{j,P}(\tau).
\end{equation}
Because the $U_{j,P}(t)$ is a non-decreasing function of $t$, and $\tau \leq \nu^{-}$, \eqref{pf27} implies
\begin{equation}\label{pf29}
s_\pi\leq U_{j,P}(\nu^{-}).
\end{equation}
Also, from \eqref{hyp1np}, we have
\begin{equation}\label{pf230}
U_{j,\pi}(\nu^{-})\leq U_{j,P}(\nu^{-}).
\end{equation}
Combining \eqref{pf29} and \eqref{pf230} with \eqref{Def1np}, we obtain
\begin{equation}
U_{j,P}(\nu)\geq U_{j,\pi}(\nu).
\end{equation}
Since there is no packet delivery under other links, we get
\begin{equation}
\begin{split}
U_{k,P}(\nu)=U_{k,P}(\nu^{-})\geq U_{k,\pi}(\nu^{-})=U_{k,\pi}(\nu), \quad \forall k\neq j.
\end{split}
\end{equation}
Hence, we have 
\begin{equation}
\mathbf{U}_{P}(\nu) \geq \mathbf{U}_{\pi}(\nu),
\end{equation}
which complete the proof.
\end{proof}

\begin{lemma}\label{lem4np}
Suppose that under policy $P$, $\mathbf{U'}_{P}$ is obtained by the arrival of a new packet to node $0$ in the system whose state is $\mathbf{U}_{P}$. Further, suppose that under policy $\pi$, $\mathbf{U'}_{\pi}$ is obtained by the arrival of a new packet to node $0$ in the system whose state is $\mathbf{U}_\pi$. If
\begin{equation}\label{hyp2np}
 \mathbf{U}_{P} \geq \mathbf{U}_\pi,
\end{equation}
then,
\begin{equation}
\mathbf{U'}_{P} \geq \mathbf{U'}_{\pi}.
\end{equation}
\end{lemma}

\begin{proof}
The proof of Lemma \ref{lem4np} is similar to that of Lemma \ref{lem4}, and hence is not provided.
\end{proof}

\begin{proof}[Proof of Lemma \ref{lem2np}]
For any sample path, we have that $\mathbf{U}_{P}(0^-) = \mathbf{U}_{\pi}(0^-)$. This, together with Lemma \ref{lem3np} and Lemma \ref{lem4np},  implies that  
\begin{equation}
\begin{split}
[\mathbf{U}_{P}(t)\vert\mathcal{I}] \geq [\mathbf{U}_{\pi}(t)\vert\mathcal{I}],\nonumber
\end{split}
\end{equation}
holds for all $t\in[0,\infty)$. Hence, \eqref{main_eqnp} holds which implies \eqref{law9np} by Theorem 6.B.30 in \cite{shaked2007stochastic}.
This completes the proof.
\end{proof}

\begin{proof}[Proof of Theorem \ref{thm2}]
According to Lemma \ref{lem2np}, we have
\begin{equation*}
\begin{split}
[\{\mathbf{U}_{P}(t),  t\in[0,\infty)\}\vert\mathcal{I}]\geq_{\text{st}} [\{\mathbf{U}_{\pi}(t), t\in[0,\infty)\}\vert\mathcal{I}],
 \end{split}
\end{equation*}
holds for all $\pi\in\Pi_{npwc}$, which implies
\begin{equation*}
\begin{split}
[\{\mathbf{\Delta}_{P}(t), t\in[0,\infty)\}\vert\mathcal{I}]\!\!\leq_{\text{st}} \!\![\{\mathbf{\Delta}_{\pi}(t), t\in[0,\infty)\}\vert\mathcal{I}],
 \end{split}
\end{equation*}
holds for all $\pi\in\Pi_{npwc}$. This completes the proof. 
\end{proof}
\bibliographystyle{IEEEbib}
\bibliography{MyLib}

\begin{IEEEbiography}
    [{\includegraphics[width=1.1in,height=1.3in]{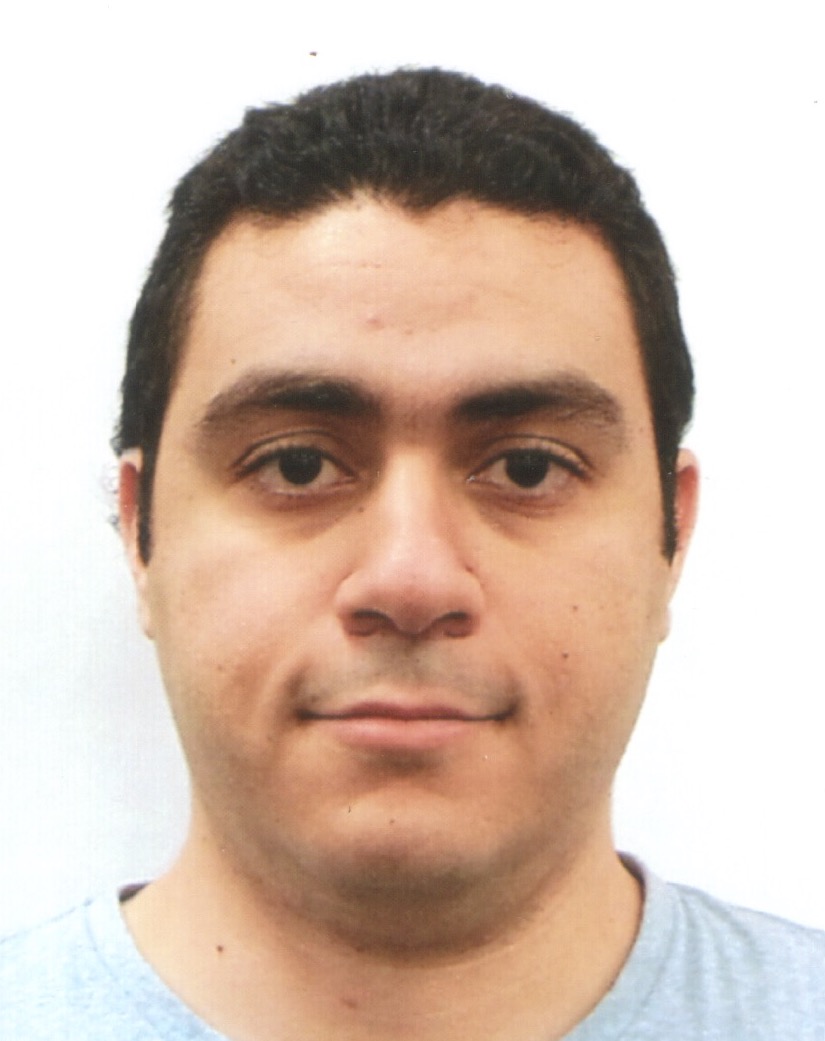}}]{Ahmed M. Bedewy}
received  the B.S.  and  M.S.  degrees  in  electrical  and  electronics engineering from Alexandria University, Alexandria,  Egypt, in 2011 and 2015,   respectively.   He   is   currently   pursuing   the Ph.D.   degree   with   the   Electrical   and   Computer Engineering Department, at the Ohio State University, OH, USA.  His  research interests  include wireless communication, cognitive radios, resource allocation, communication networks, information freshness, optimization, and scheduling algorithms.   He awarded the Certificate of Merit, First Class Honors, for being one of the top ten undergraduate students during 2006-2008 and for being \textbf{1-st} during 2008-2011 in electrical and electronics engineering.
\end{IEEEbiography}

\begin{IEEEbiography}
    [{\includegraphics[width=1.1in,height=1.3in]{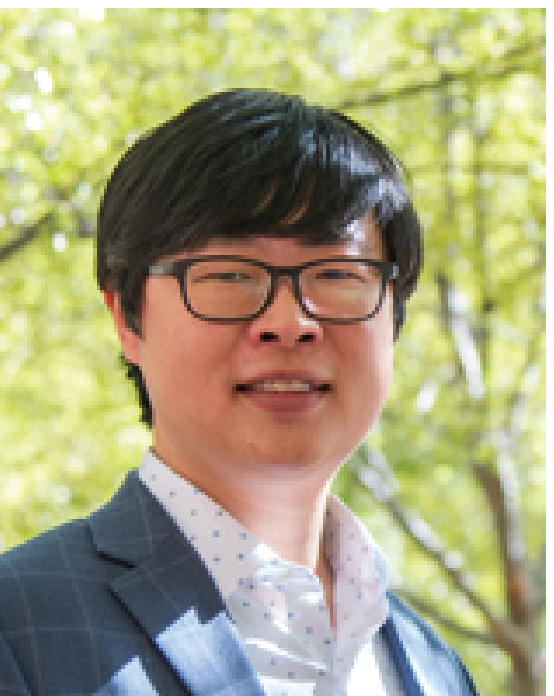}}]{Yin Sun} (S'08-M'11) received his B.Eng. and Ph.D. degrees in Electronic Engineering from Tsinghua University, in 2006 and 2011, respectively. At Tsinghua, he received the Excellent Doctoral Thesis Award of Tsinghua University, among many awards and scholarships. He was a postdoctoral scholar and research associate at the Ohio State University during 2011-2017. Since Fall 2017, Dr. Sun joined Auburn University as an assistant professor in the Department of Electrical and Computer Engineering.

His research interests include wireless communications, communication networks, information freshness, information theory, and machine learning. He is the founding co-chair of the first and second Age of Information Workshops, in conjunction with the IEEE INFOCOM 2018 and 2019. The paper he co-authored received the best student paper award at IEEE WiOpt 2013.
\end{IEEEbiography}

\begin{IEEEbiography}
    [{\includegraphics[width=1.1in,height=1.3in]{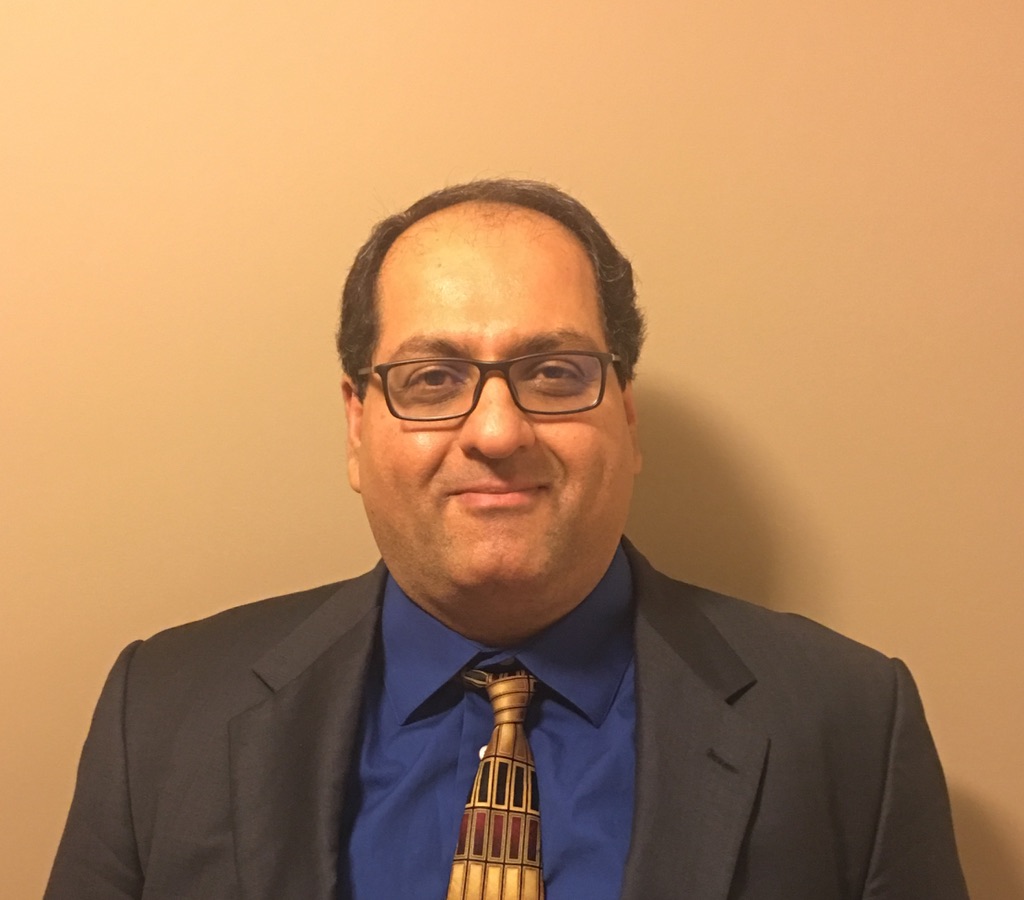}}]{Ness B. Shroff} (S’91–M’93–SM’01–F’07)
received the Ph.D. degree in electrical engineering from Columbia University in 1994. He joined Purdue University immediately thereafter as an Assistant Professor with the School of Electrical and Computer Engineering. At Purdue, he became a Full Professor of ECE and the director of  a university-wide center on wireless systems and applications in 2004. In 2007, he joined The Ohio State University, where he holds the Ohio Eminent Scholar Endowed Chair in networking and communications, in the departments of ECE and CSE. He holds or has held visiting (chaired) professor positions at Tsinghua University, Beijing, China, Shanghai Jiaotong University, Shanghai, China, and IIT Bombay, Mumbai, India. He has received numerous best paper awards for his research and is listed in Thomson Reuters’ on The World’s Most Influential Scientific Minds, and is noted as a Highly Cited Researcher by Thomson Reuters. He also received the IEEE INFOCOM Achievement Award for seminal contributions to scheduling and resource allocation in wireless networks. He currently serves as the steering committee chair for ACM Mobihoc and Editor at Large of the IEEE/ACM Transactions on Networking. 
\end{IEEEbiography}

\end{document}